\documentclass[a4paper,UKenglish,cleveref, autoref, thm-restate]{lipics-v2021}
\pdfoutput=1 
\hideLIPIcs  

\usepackage{color}
\usepackage{textgreek}
\usepackage{amsmath,amssymb}
\usepackage{amsthm}
\usepackage{thm-restate}
\usepackage{hyperref}

\newtheorem{assumption}[theorem]{Assumption}
\newenvironment{sketch}{
\renewcommand{\proofname}{Proof Sketch}
\renewcommand{\qed}{}
\proof}{\endproof}

\newcommand{\ie}{\textit{i.e.,~}}

\newcommand{\etal}{\textit{et~al.~}}

\usepackage[linesnumbered,ruled,vlined,noresetcount]{algorithm2e}
\DontPrintSemicolon
\SetKwProg{Fn}{function}{:}{end}
\SetKwProg{when}{when }{ \textbf{do}}{done}
\SetKwProg{For}{for }{ \textbf{do}}{done}
\SetKwProg{Forall}{forall }{ \textbf{do}}{endfor}
\let\oldnl\nl
\newcommand{\nonl}{\renewcommand{\nl}{\let\nl\oldnl}}

\newcommand{\Call}{\mathcal{C}_{all}}
\newcommand{\EHT}{\mathrm{EHT}}
\newcommand{\ECT}{\mathrm{ECT}}

\newcommand{\PIDS}{{\mathcal{P}_{\mathrm{ID2}}}}

\newcommand{\PRDS}{{\mathcal{P}_{\mathrm{RD2}}}}

\newcommand{\deterministic}{deterministic}
\newcommand{\Deterministic}{Deterministic}

\newcommand{\clr}{\mathtt{color}}
\newcommand{\NC}{\mathcal{NC}}
\newcommand{\PNC}{{\mathcal{P}_{\mathrm{NC}}}}

\newcommand{\PLRU}{{\mathcal{P}_{\mathrm{LRU}}}}
\newcommand{\PDLRU}{{\mathcal{P'}_{\mathrm{LRU}}}}
\newcommand{\khopin}{{8N^3\Delta^2}}
\newcommand{\khop}{{\lceil\log{\khopin}\rceil}}
\newcommand{\randsize}{{\Delta}}
\newcommand{\nc}{\mathtt{normalcolor}}
\newcommand{\rand}{\mathtt{rand}}
\newcommand{\hop}{\mathtt{hopcolor}}
\newcommand{\prev}{\mathtt{prev}}
\newcommand{\stamp}{\mathtt{stamp}}
\newcommand{\idx}{\mathtt{idx}}
\newcommand{\cur}{\mathtt{cursor}}
\newcommand{\gencur}{\mathtt{gencur}}
\newcommand{\gen}{\mathtt{gen}}
\newcommand{\Ncc}{{\mathcal{N}_{\mathrm{cc}}}}
\newcommand{\Scol}{{\mathcal{S}_{\mathrm{color}}}}

\newcommand{\genrandom}[3]{generate #1 $#2 \in #3$ uniformly at random}

\newcommand{\pcol}{\mathtt{pcol}}
\newcommand{\LF}{\mathtt{LF}}
\newcommand{\LL}{\mathtt{L}}
\newcommand{\FF}{\mathtt{F}}
\newcommand{\iid}{\mathtt{id}}
\newcommand{\timer}{\mathtt{timer}}
\newcommand{\KL}{\mathtt{KL}}
\newcommand{\tbc}{\mathtt{t_{\mathrm{BC}}}}
\newcommand{\Baby}{\mathtt{B}}
\newcommand{\rc}{\mathtt{rc}}
\newcommand{\type}{\mathtt{type}}
\newcommand{\PBC}{{\mathcal{P}_{\mathrm{BC}}}}

\newcommand{\SDcol}{{\mathcal{S'}_{\mathrm{color}}}}
\newcommand{\SLE}{{\mathcal{S}_{\mathrm{LE}}}}
\newcommand{\KLzero}{{\mathcal{KL}_{\mathrm{zero}}}}
\newcommand{\KLhalf}{{\mathcal{KL}_{\mathrm{half}}}}
\newcommand{\Bexists}{{\mathcal{B}_{\mathrm{exists}}}}
\newcommand{\Lone}{{\mathcal{L}_{\mathrm{one}}}}
\newcommand{\Fall}{{\mathcal{F}_{\mathrm{all}}}}
\newcommand{\LFqua}{{\mathcal{LF}_{\mathrm{qua}}}}
\newcommand{\Ldupl}{{\mathcal{L}_{\mathrm{dupl}}}}
\newcommand{\Lexists}{{\mathcal{L}_{\mathrm{exists}}}}

\newcommand{\Lvzero}{{\mathcal{L}_{\mathrm{v0}}}}
\newcommand{\Lvone}{{\mathcal{L}_{\mathrm{v1}}}}
\newcommand{\Vclean}{{\mathcal{V}_{\mathrm{clean}}}}
\newcommand{\Vmake}{{\mathcal{V}_{\mathrm{make}}}}
\newcommand{\Vonly}{{\mathcal{V}_{\mathrm{only}}}}
\newcommand{\Ehalf}{{\mathcal{E}_{\mathrm{half}}}}
\newcommand{\Equa}{{\mathcal{E}_{\mathrm{qua}}}}
\newcommand{\Bno}{{\mathcal{B}_{\mathrm{no}}}}
\newcommand{\IDclear}{{\mathcal{ID}_{\mathrm{clear}}}}
\newcommand{\IDsame}{{\mathcal{ID}_{\mathrm{same}}}}

\title{Almost Time-Optimal Loosely-Stabilizing Leader Election on Arbitrary Graphs Without Identifiers in Population Protocols} 

\titlerunning{Almost Time-Optimal Loosely-Stabilizing Leader Election on Graphs Without ID} 

\author{Haruki Kanaya}{Nara Institute of Science and Technology, Japan}{kanaya.haruki.kk3@naist.ac.jp}{https://orcid.org/0009-0003-8102-0280}{}
\author{Ryota Eguchi}{Nara Institute of Science and Technology, Japan}{ry.eguchi@is.naist.jp}{https://orcid.org/0000-0002-4836-2903}{}
\author{Taisho Sasada}{Nara Institute of Science and Technology, Japan}{taisho.sasada@naist.ac.jp}{https://orcid.org/0000-0003-2144-4949}{}
\author{Michiko Inoue}{Nara Institute of Science and Technology, Japan}{kounoe@is.naist.jp}{https://orcid.org/0000-0002-9837-5147}{}

\authorrunning{H. Kanaya, R. Eguchi, T. Sasada, and M. Inoue}

\Copyright{Haruki Kanaya, Ryota Eguchi, Taisho Sasada, and Michiko Inoue}



\funding{The second author was supported by JSPS KAKENHI Grant Number 22H03569.}

\nolinenumbers 

\begin{document}

\maketitle

\begin{abstract}
The population protocol model is a computational model for passive mobile agents. 
We address the leader election problem, which determines a unique leader on arbitrary communication graphs starting from any configuration.
Unfortunately, self-stabilizing leader election is impossible to be solved without knowing the exact number of agents; thus, we consider loosely-stabilizing leader election, which 
converges to safe configurations in a relatively short time, and holds the specification (maintains a unique leader)
for a relatively long time. 
When agents have unique identifiers, Sudo~\etal(2019) proposed a protocol that, given an upper bound \(N\) for the number of agents \(n\), converges in \(O(mN\log n)\) expected steps, where \(m\) is the number of edges. 
When unique identifiers are not required,
they also proposed a protocol that, using random numbers and given $N$, converges in $O(mN^2\log{N})$ expected steps.
Both protocols have a holding time of $\Omega(e^{2N})$ expected steps and use $O(\log{N})$ bits of memory.
They also showed that the lower bound of the convergence time is \(\Omega(mN)\) expected steps for protocols with a holding time of \(\Omega(e^N)\) expected steps given \(N\).

In this paper, we propose protocols that do not require unique identifiers. These protocols achieve convergence times close to the lower bound with increasing memory usage.
Specifically, given \(N\) and an upper bound \(\Delta\) for the maximum degree, we propose two protocols whose convergence times are \(O(mN\log n)\) and \(O(mN\log N)\) both in expectation and with high probability. The former protocol uses random numbers, while the latter does not require them. Both protocols utilize \(O(\Delta \log N)\) bits of memory and hold the specification for \(\Omega(e^{2N})\) expected steps.
\end{abstract}

\section{Introduction}
The population protocol model, introduced by Angluin~\etal\cite{AngluinADFP2006}, is a computational model widely recognized in distributed computing and applicable to passive mobile sensor networks, chemical reaction systems, and molecular calculations, etc. 
This model comprises \(n\) finite state machines (called \emph{agents}), which form a network (called a \emph{population}). 
Agents' states are updated through communication (called \emph{interaction}) among a pair of agents. 
A simple connected digraph \(G=(V,E)\) (called a \emph{communication graph}) determines the possibility of interaction among the agents. 
In this model, only one pair of agents interacts at each step. 
The interactions are determined by a uniform random scheduler.

The leader election problem is one of the most studied problems in population protocols. 
This problem involves agents electing a unique \emph{leader} agent from the population and maintaining this unique leader forever. 
Angluin~\etal\cite{AngluinADFP2006} first studied this problem for complete graphs with designated common initial state.
Under this assumption, many studies have been conducted
~\cite{AngluinADFP2006,BGK20,GSU19,sudotime}, and a time and space optimal protocol~\cite{BGK20} has already been proposed. 
Several studies also exist for arbitrary graphs~\cite{FastGraphical,NearArbGraph}, and a time-optimal protocol~\cite{NearArbGraph} has already been proposed.

The self-stabilizing leader election problem requires that agents start from any configuration, elect and externally maintain a unique leader agent. 
It is known that there is no self-stabilizing leader election protocol for arbitrary graphs~\cite{SSAngluin} and complete graphs~\cite{Cai2012}
, and researchers have explored the problem in three  ways. 
The first approach involves assuming that all agents initially know the exact number $n$ of agents~\cite{SSTimeOp,Cai2012,SudoGlobal}. 
The second approach introduces an oracle that informs agents about the existence of leaders~\cite{BeauquierOracle,canepa,FischerOracle}. 
The third approach relaxes the requirement of maintaining a unique leader forever, introducing a \emph{loosely-stabilizing} leader election problem, where agents start from any configuration, elect a unique leader within a short time, and maintain this leader for a long time.
Sudo~\etal\cite{SUDOLSLE} first addressed  this problem on complete graphs. 
Subsequent studies have continued to explore this problem~\cite{OnSpaceIzumi,SudoLSLETime,SudoArb1st,SUDOArb3rd,SameSpeedTimer,SUDOpoly}
as follows and summarized in Table~\ref{table:liststudies}.

Sudo, Ooshita, Kakugawa, and Masuzawa~\cite{SudoArb1st} first addressed this problem for arbitrary graphs, and it is significantly improved by Sudo, Ooshita, Kakugawa, Masuzawa, Datta, and Lawrence~\cite{SameSpeedTimer} introducing a novel concept of Same Speed Timer. 
They proposed two protocols.
The first protocol, \(\PIDS\)
, assumes that agents have unique identifiers and are given \(N\) as initial knowledge. 
\(\PIDS\) converges within \(O(mN\log{n})\) expected steps and holds the unique leader with \(\Omega(Ne^{2N})\) expected steps, using \(O(\log{N})\) bits of memory. 
The second protocol, \(\PRDS\), assumes that agents can make randomized transitions and is given \(N\) as initial knowledge. 
\(\PRDS\) converges within \(O(mN^2\log{N})\) expected steps and holds the unique leader with \(\Omega(Ne^{2N})\) expected steps using \(O(\log{N})\) bits of memory.
Sudo~\etal also demonstrated that the lower bound of the convergence time is $\Omega(mN)$ steps for any loosely-stabilizing protocols with holding a unique leader $\Omega(e^N)$ expected steps. 

Loosely-stabilizing leader election protocols without requiring unique identifiers or random numbers were proposed~\cite{SUDOArb3rd} and then improved~\cite{SudoGlobal}. 
The protocol, \(\mathcal{P}_{\mathrm{AR}}\), given \(N\) and \(\Delta\) as initial knowledge, converges within \(O(mnD\log{n}+mN\Delta^2\log{N})\) expected steps and holds with \(\Omega(Ne^{N})\) expected steps using \(O(\log{N})\) bits of memory~\cite{SudoGlobal}.

\begin{table}[!tb]
    \caption{
    Convergence and Holding Times for Loosely-Stabilizing Leader Election Protocols with Exponential Holding Times.
    $n$ denotes the number of agents, $N$ denotes the upper bound of $n$, $m$ denotes the number of edges of the communication graph, $D$ denotes the diameter of the communication graph, and $\Delta$ denotes the upper bound of the maximum degree of the communication graph.
    All protocols are given $N$ as initial knowledge. 
    Protocols with $\ast$ are also given $\Delta$.
    The symbol $\dagger$ represents lower bounds of convergence time or memory usage for protocols with holding time of $\Omega(e^N)$.
    }
    \label{table:liststudies}
    \centering
    \begin{tabular}{ccccccc}
        \hline
         & Graph & Convergence & Holding & Memory & Requisite\\
        \hline
        \cite{SUDOLSLE} & complete & $O(nN\log{n})$ & $\Omega(Ne^N)$ & $O(\log{N})$  & -\\
        \cite{OnSpaceIzumi} & complete & $O(nN)$ & $\Omega(e^N)$ & $O(\log{N})$  & - \\
        \cite{OnSpaceIzumi}$\dagger$ & complete & $\Omega(nN)$ & $\Omega(e^N)$ & - & -\\
        \cite{OnSpaceIzumi}$\dagger$ & complete & - & $\Omega(e^N)$ & $\Omega(\log{N})$ & -\\
        \hline
        \cite{SudoArb1st}$\ast$ & arbitrary & $O(m\Delta N\log{n})$ & $\Omega(Ne^N)$  & $O(\log{N})$  & agent identifiers\\
        \cite{SudoArb1st}$\ast$ & arbitrary & $O(m\Delta^2 N^3\log{N})$ & $\Omega(Ne^N)$  & $O(\log{N})$ & random numbers\\
        \cite{SameSpeedTimer} & arbitrary & $O(mN\log{n})$ & $\Omega(Ne^{2N})$  & $O(\log{N})$ & agent identifiers\\
        \cite{SameSpeedTimer} & arbitrary & $O(mN^2\log{N})$ & $\Omega(Ne^{2N})$  & $O(\log{N})$ & random numbers\\
        \cite{SameSpeedTimer}$\dagger$ & arbitrary & $\Omega(mN)$ & $\Omega(e^N)$ & - & -\\
        \cite{SUDOArb3rd}$\ast$ & arbitrary & $O(mnD\log{n}+mN\Delta^2\log{N})$ & $\Omega(Ne^N)$ & $O(\log{N})$  & -\\
        This$\ast$ & arbitrary & $O(mN\log{n})$ & $\Omega(Ne^{2N})$ & $O(\Delta\log{N})$ & random numbers\\
        This$\ast$ & arbitrary & $O(mN\log{N})$ & $\Omega(Ne^{2N})$ & $O(\Delta\log{N})$ & -\\
        \hline
    \end{tabular}
\end{table}

\subsection{Our Contribution}
In this paper, we propose a protocol \(\PBC\) whose convergence time is nearly optimal on anonymous (without unique identifiers) arbitrary graphs, as supported by the lower bound~\cite{SameSpeedTimer}. 
The lower bound for the complete graph~\cite{OnSpaceIzumi} is \(\Omega(nN)\), but this is a special case, as the lower bound~\cite{SameSpeedTimer} is known to hold even when \(m = \Theta(n^2)\).
Given  
\(N\) and \(\Delta\), \(\PBC\) converges within \(O(mN\log n)\) steps if the transition is randomized, and \(O(mN\log N)\) steps if the transition is \deterministic\footnote{
Transition is said to be \textit{deterministic} if it does not require random numbers in the transition. Though the transition is deterministic, we allow the protocol to exploit the randomness with which initiator and responder roles are chosen. 
}, both in expectations and with high probability. 
The protocol holds the unique leader with \(\Omega(Ne^{2N})\) expected steps and utilizes \(O(\Delta\log N)\) bits of memory.
The proposed $\PBC$ has better convergence time than SOTA self-stabilizing leader election protocol~\cite{SudoGlobal} which converges with $O(mn^2D\log{n})$ steps with requiring the knowledge of $n$.

To achieve the convergence time of \(\PBC\), we utilize the Same Speed Timer proposed in \(\PRDS\)~\cite{SameSpeedTimer}, which requires two-hop coloring.
The \emph{self-stabilizing two-hop coloring} protocol was first studied by Angluin~\etal\cite{SSAngluin}, and further explored by Sudo~\etal\cite{SameSpeedTimer} (see Table~\ref{table:2hop}). 
In this paper, we propose two new self-stabilizing two-hop coloring protocols; \(\PLRU\) with randomized transitions, and \(\PDLRU\) with \deterministic~transitions. 
Both protocols require \(N\) and \(\Delta\) as initial knowledge. 
\(\PLRU\) converges within \(O(mn)\) steps, both in expectation and with high probability, and uses \(O(\Delta\log{N})\) bits of memory. 
\(\PDLRU\) converges within \(O(m(n+\Delta\log{N}))\) steps, both in expectation and with high probability, and also uses \(O(\Delta\log{N})\) bits of memory.
In \(\PDLRU\), agents generate random numbers independently from interactions among themselves.
To ensure the independence among random numbers, we employ the \emph{self-stabilizing normal coloring} protocol \(\PNC\) to assign superiority or inferiority between adjacent agents. 
When interacting, only the superior agent uses the interaction to generate random numbers. \(\PNC\) converges within \(O(mn\log{n})\) steps, both in expectation and with high probability, and utilizes \(O(\log{N})\) bits of memory.

\begin{table}[!hbt]
    \caption{List of Convergence Times for Self-Stabilizing Two-Hop Coloring Protocols on Arbitrary Graphs. $n$ denotes the number of agents, $N$ denotes the upper bound of $n$, $m$ denotes the number of edges of the communication graph, $\delta$ denotes the maximum degree of the communication graph, and $\Delta$ denotes the upper bound of $\delta$.}
    \label{table:2hop}
    \centering
    \begin{tabular}{ccccc}
        \hline
         & Convergence & Memory & Knowledge & Requisite\\
        \hline
        Angluin~\etal\cite{SSAngluin} & - & $O(\Delta^2)$ & $\Delta$ & random numbers\\
        Angluin~\etal\cite{SSAngluin} & - & $O(\Delta^2)$ & $\Delta$ & -\\
        Sudo~\etal\cite{SameSpeedTimer} & $O(mn\delta\log{n})$ & $O(\log{N})$ & $N$ & random numbers\\
        $\PLRU$ (this) & $O(mn)$ & $O(\Delta\log{N})$ & $N,\Delta$ & random numbers\\
        $\PDLRU$ (this) & $O(mn+m\Delta\log{N})$ & $O(\Delta\log{N})$ & $N,\Delta$ & -\\
        \hline
    \end{tabular}
\end{table}

\section{Preliminaries}
In this paper, \(\mathbb{N}\) denotes the set of natural numbers no less than one, and \(\log{x}\) refers to \(\log_2{x}\). 
If we use the natural logarithm, we explicitly specify the base \(e\) by writing \(\log_e{x}\).

A population is represented by a simple connected digraph \(G=(V,E)\), where \(V\) (\(|V| \ge 2\)) represents a set of agents, and \(E \subseteq \{(u,v) \in V \times V \mid u \neq v\}\) represents the pairs of agents indicating potential interactions. 
An agent \(u\) can interact with an agent \(v\) if and only if \((u,v) \in E\), where \(u\) is the initiator and \(v\) is the responder.
We assume $G$ is symmetric, that is, if for every \((u,v) \in V \times V\), the preposition \((u,v) \in E \Rightarrow (v,u) \in E\) holds. 
We also denote \(n = |V|\) and \(m = |E|\).
The diameter of $G$ is denoted by $D$.
The degree of agent \(u\) is denoted by \(\delta_u = |\{v \in V \mid (u,v) \in E \vee (v,u) \in E\}|\), and the maximum degree is denoted by \(\delta = \max_{u \in V} \{\delta_u\}\). 
The upper bound \(N\) of \(n\) satisfies \(N \geq n\), and the upper bound \(\Delta\) of \(\delta\) satisfies \(\delta \leq \Delta \leq 2(N-1)\) \footnote{
Given only \(N\), \(\delta \leq 2(N-1)\) holds, thus \(\Delta \leq 2(N-1)\).
}.

A protocol \(\mathcal{P}\) is defined as a 5-tuple \((Q, Y, R, T, O)\), where \(Q\) represents the finite set of states of agents, \(Y\) represents the finite set of output symbols, \(R \subset \mathbb{N}\) represents the range of random numbers, \(T: Q \times Q \times R \rightarrow Q \times Q\) is the transition function, and \(O: Q \rightarrow Y\) is the output function. 
When an initiator \(u\), whose state is \(p \in Q\), interacts with a responder \(v\), whose state is \(q \in Q\), each agent updates their states via the transition function using their current states and a random number \(r \in R\) to \(p', q' \in Q\) such that \((p', q') = T(p, q, r)\). 
An agent whose state is \(p \in Q\) outputs \(O(p) \in Y\). 
A protocol \(\mathcal{P}\) is with \deterministic~transitions if and only if \(\forall r, \forall r' \in R, \forall p, \forall q \in Q : T(p, q, r) = T(p, q, r')\) holds.
Otherwise, a protocol \(\mathcal{P}\) is with randomized transitions.
The memory usage of protocol \(\mathcal{P}\) is defined by \(\lceil \log{|Q|} \rceil\) bits.

A \emph{configuration} \(C: V \rightarrow Q\) represents the states of all agents. 
The set of all configurations by protocol \(\mathcal{P}\) is denoted by \(\Call(\mathcal{P})\). 
A configuration \(C\) transitions to \(C'\) by an interaction \(e=(u,v)\) and a random number \(r \in R\) if and only if \((C'(u), C'(v)) = T(C(u), C(v), r)\) and \(\forall w \in V \setminus \{u, v\}: C'(w) = C(w)\) holds. 
Transitioning from a configuration \(C\) to \(C'\) by an interaction \(e\) and a random number \(r\) is denoted by \(C \xrightarrow{e, r} C'\).
A uniform random scheduler \(\Gamma = \Gamma_0, \Gamma_1, \dots\) determines which pair of agents interact at each step, where \(\Gamma_t \in E\) (for \(t \geq 0\)) is a random variable satisfying \(\forall (u,v) \in E, \forall t: \Pr(\Gamma_t = (u,v)) = 1/m\).
An infinite sequence of random numbers \(\Lambda = R_0, R_1, \dots\) represents a random number generated at each step, where \(R_t\) (for \(t \geq 0\)) is a random variable satisfying \(\forall r \in R, \forall t: \Pr(R_t = r) = 1/|R|\).
Given an initial configuration \(C_0 \in \Call(\mathcal{P})\), a uniform random scheduler \(\Gamma\), and a sequence of random numbers \(\Lambda\), the execution of protocol \(\mathcal{P}\) is denoted by \(\Xi_{\mathcal{P}}(C_0, \Gamma, \Lambda)=C_0,C_1,\dots\) where \(C_t \xrightarrow{\Gamma_t, R_t} C_{t+1}\) (for \(t \geq 0\)) holds. 
If \(\Gamma\) and \(\Lambda\) are clear from the context, we may simply write \(\Xi_{\mathcal{P}}(C_0)\).
A set of configurations $\mathcal{S}$ is safe if and only if there is no configuration $C_i\notin \mathcal{S}\ (i\in \mathbb{N})$ for any configuration $C_0\in \mathcal{S}$ and any execution $\Xi(C_0)=C_0,C_1,\dots$.
A protocol is silent if and only if there is no state changed after reached safe configurations.

For a protocol \(\mathcal{P}\) that solves a population protocol problem, the expected holding time and the expected convergence time are defined as follows. 
The specification of the problem, which is a required condition for an execution, is denoted by \(\mathscr{SC}\). 
For any configuration \(C \in \Call(\mathcal{P})\), any uniform random scheduler $\Gamma$, and any infinite sequence of random numbers $\Lambda$, the expected number of steps that an execution \(\Xi_\mathcal{P}(C,\Gamma,\Lambda)\) satisfies \(\mathscr{SC}\) is defined as the expected holding time, denoted \(\EHT_\mathcal{P}(C, \mathscr{SC})\). 
For any set of configurations \(\mathcal{S} \subseteq \Call(\mathcal{P})\), any configuration \(C \in \Call(\mathcal{P})\), any uniform random scheduler $\Gamma$, and any infinite sequence of random numbers $\Lambda$, the expected number of steps from the beginning of the execution \(\Xi_\mathcal{P}(C,\Gamma,\Lambda)\) until the configuration reaches \(\mathcal{S}\) is defined as the expected convergence time, denoted \(\ECT_\mathcal{P}(C, \mathcal{S})\).
A computation is considered to be finished with high probability if and only if the computation finishes with probability \(1 - O(n^{-c})\) for \(c \geq 1\).

The leader election problem requires that all agents output either \(L\) or \(F\), where \(L\) represents a leader and \(F\) represents a follower. 
The specification of the leader election is denoted by \(LE\). 
For an execution \(\Xi_\mathcal{P}(C_0) = C_0, C_1, \dots, C_x, \dots\), the configurations \(C_0, \dots, C_x\) satisfy \(LE\) if and only if there is an agent \(u\) such that \(\forall i \in [0, x]: O(C_i(u)) = L\), and \(\forall i \in [0, x], \forall v \in V \setminus \{u\}: O(C_i(v)) = F\) holds.

\begin{definition}[Loosely-stabilizing leader election\cite{SUDOLSLE}]
A protocol $\mathcal{P}$ is an $(\alpha,\beta)$-loosely-stabilizing leader election protocol if and only if there exists a set of configurations $\mathcal{S} \subseteq \Call(\mathcal{P})$ such that $\max_{C \in \Call(\mathcal{P})} \ECT_\mathcal{P}(C, \mathcal{S}) \leq \alpha$ and $\min_{C \in \mathcal{S}} \EHT_\mathcal{P}(C, LE) \geq \beta$ holds.
\end{definition}

A protocol $\mathcal{P}$ is a self-stabilizing protocol of a problem if and only if there exists safe configurations that any execution starting from any safe configuration satisfies the specification of the problem (called closure), and any execution starting from any configuration includes a safe configuration reaches the safe configurations (called convergence).

\begin{definition}
A protocol $\mathcal{P}$ is a self-stabilizing normal coloring protocol if and only if there exists non-negative integer $x$ such that for any configuration $C_0\in \Call(\mathcal{P})$, and the execution $\Xi_\mathcal{P}(C_0)=C_0,C_1,\dots,C_x,\dots$, the following condition holds:
$\forall i\in \mathbb{N},\forall v\in V:O(C_{x}(v))=O(C_{x+i}(v))$ and $\forall v,\forall u\in V:(u,v)\in E \Rightarrow O(C_x(u))\ne O(C_x(v))$.
\end{definition}

\begin{definition}
A protocol $\mathcal{P}$ is a self-stabilizing two-hop coloring protocols if and only if there exists non-negative integer $x$ such that for any configuration $C_0\in \Call(\mathcal{P})$, and the execution $\Xi_\mathcal{P}(C_0)=C_0,C_1,\dots,C_x,\dots$, the following condition holds:
$\forall i\in\mathbb{N},\forall v\in V:O(C_{x}(v))=O(C_{x+i}(v))$ and $\forall v,\forall u,\forall w\in V:(u,v)\in E \wedge (v,w)\in E  \Rightarrow O(C_x(u))\ne O(C_x(w))$.
\end{definition}

\section{Self-Stabilizing Two-Hop Coloring}
\label{sec:twohopmain}
In this section, we introduce a 
self-stabilizing two-hop coloring protocol with randomized transitions $\PLRU$, alongside a \deterministic~self-stabilizing two-hop coloring protocol with \deterministic~transitions $\PDLRU$.

Two distinct agents $u,v\in V$ are called two-hop located if and only if there exists $w \in V$ such that $(u,w) \in E \wedge (v,w) \in E$. 
A graph is considered two-hop colored if and only if, for any pair of agents $u$ and $v$ that are two-hop located, $u$ and $v$ are assigned distinct colors. 

The basic strategy is similar to that described by Angluin~\etal\cite{SSAngluin} and Sudo~\etal\cite{SameSpeedTimer}. 
The differences lie in the methods for generating colors and the length of the array used to record the colors of interacted agents. Angluin~\etal memorized all generated colors using an array of length $\Delta(\Delta-1)+1$, whereas Sudo~\etal recorded only the most recent color. In the protocols, the agents record the last $\Delta$ colors.

We present a general strategy for color collision detection.
When interacting two agents, they record each other's color with a common binary random stamp. If there is no color collision, when they interact again they find that they remember each other's color with the same stamp value. Assume that agents $v$ and $w$ have a color collision, that is, they are two-hop located with a common neighbor $u$ and have the same color. Consider the scenario in which interactions occur in the order of $(u,v)$, $(u,w)$, $(u,v)$, and $u$ (resp. $v$) records $v$'s (resp. $u$'s) color with a stamp $0$ and then $u$ records $w$'s color (it is also $v$'s color) with a stamp $1$. When $u$ and $v$ interact again, they notice they remember each other's color with different stamp values and detect the color collision.

In both protocols, each agent has arrays whose size are $\Delta$ to record the last $\Delta$ colors and their stamps.
Both protocols are the same except for the way to generate colors.
In $\PLRU$, agents generate colors by using the ability of generating  uniform random numbers.
In $\PDLRU$, agents generate colors by using the roles of initiator and responder.
To generate $x$-digit binary random number, each agent generates a one random bit according to its role (initiator or responder) in each interaction, and repeats it $x$ times.
To ensure the independence of random numbers, only one agent can use the interaction to generate random numbers for each interaction.
To solve this issue, we use normal coloring.
A graph is considered normal colored if and only if agents $u$ and $v$ are different colors for any pair $(u,v)\in E$.
We call agents' colors which are colored by normal coloring the \emph{normal color}.
To guarantee independency among random numbers, when two agents interact, the agent with larger normal color value can use the interaction to generate random numbers.
Though this mechanism does not give a chance to generate random numbers to agents with smaller normal color, random numbers are used when two agents detect a color collision.
In such cases, two random numbers are provided as new colors from an agent who has larger normal color value and already generated two or more numbers.

\subsection{Protocol \texorpdfstring{$\PLRU$}{P\_{LRU}}}
In this subsection, we introduce the randomized self-stabilizing two-hop coloring protocol $\PLRU$. 
Given $N$ and $\Delta$, the protocol $\PLRU$ achieves convergence within $O(mn)$ steps, both in expectation and with high probability, while requiring $O(\Delta \log N)$ bits of memory per agent.

An agent $a$ in $\PLRU$ has four variables: $a.\hop \in \{1, \dots, \khopin\}$, $a.\prev \in \{1, \dots, \\\khopin\}^\Delta$, $a.\stamp \in \{0,1\}^\Delta$, and $a.\idx \in \{0,\dots,\Delta\}$. The variable $a.\hop$ represents the two-hop color of the agent. The variable $a.\prev$ is an array that stores the last $\Delta$ colors interacted by the agent. The variable $a.\stamp$ is an array of size $\Delta$, with each entry being either $0$ or $1$, used to record the stamp associated with each color memorized. Lastly, $a.\idx$ serves as a temporary index to locate the color of the interacting agent in $a.\prev$.

\begin{algorithm}[!tbh]
    \caption{Self-Stabilizing two-hop coloring $\mathcal{P}_{LRU}$}
    \label{Protocol:ss2hoprandom}
    \SetKwInOut{OutFunc}{Outout Function $O$}
    \SetKwProg{Fn}{function}{:}{}
    \SetKwFunction{GenClr}{Generate\_Color}
    \SetKwFunction{GenBit}{Generate\_Bit}
    \SetKwFunction{GetRandomBit}{GetBit}
    \SetKw{Downto}{downto}
    \setcounter{AlgoLine}{0}
    \OutFunc{Each agent $a$ outputs $a.\hop$.}
    \nonl\when{an initiator $a_0$ interacts with a responder $a_1$}{
        \Forall{$i \in \{0,1\}$}{
            $a_i.\idx \gets 0$\;
            \For{$j \gets 1\ \KwTo\ \Delta$}{
                \If{$a_i.\prev[j] = a_{1-i}.\hop$}{
                    $a_i.\idx\gets j$\;
                    \textbf{break}\;
                }
            }
        }
        \GenClr()\;
        \Forall{$i \in \{0,1\}$}{
            \lIf{$a_i.\idx =0$}{
                $a_i.\idx \gets \Delta$
            }
            \For{$j \gets a_i.\idx-1\ \Downto\ 1 $}{
                $(a_i.\prev[j+1],a_i.\stamp[j+1])\gets (a_i.\prev[j],a_i.\stamp[j])$\;
            }
        }
        \GenBit()\;
    }
    \nonl\Fn{\GenClr{}}{
        \If{$a_0.\idx>0 \wedge a_1.\idx>0 \wedge a_0.\stamp[a_0.\idx] \ne a_1.\stamp[a_1.\idx])$}{
                \genrandom{two colors}{c_0,c_1}{\{1,...,8N^3\Delta^2\}}\;
                $(a_0.\hop,a_1.\hop)\gets (c_0,c_1)$\;
                $a_0.\idx \gets a_1.\idx \gets 0$\;
            }
    }
    \nonl\Fn{\GenBit{}}{
        \genrandom{bit}{b}{\{0,1\}}\;
        $(a_0.\prev[1],a_1.\prev[1],a_0.\stamp[1],a_1.\stamp[1]) \gets (a_1.\hop,a_0.\hop,b,b)$\;
    }
\end{algorithm}

$\PLRU$ is given by algorithm~\ref{Protocol:ss2hoprandom}, and consists of four parts:
i) reading memory, ii) collision detection, iii) saving colors, and iv) stamping.

i) Reading memory (lines 1--6) aims to find a color of the interacting partner in an array of recorded colors.
For each agent $a_i$ (where $i \in \{0,1\}$) interacting with another agent $a_{1-i}$, $a_i$ searches $a_i.\prev$ for $a_{1-i}.\hop$ and records the minimum index in $a_i.\idx$ if exists, otherwise, sets $a_i.\idx$ as $0$.

ii) Collision detection (lines 7, and 13--16) aims to generate new colors when a stamp collision is detected. 
To address this, two uniform random numbers are generated from the range $[1, \khopin]$. These numbers are then used to update $a_0.\hop$ and $a_1.\hop$ respectively.

iii) Saving colors (lines 8--11) aims to maintain arrays $\prev$ and $\stamp$ in a Least Recently Used (LRU) fashion.

iv) Stamping (lines 12, and 17--18) aims to generate a common binary stamp to two interacting agents, and to move a color and a stamp of this current interacting partner to the heads of arrays $\prev$ and $\stamp$.

We have following theorems. (See Appendix for the details, proofs, and implementations).

\begin{restatable}{theorem}{twohoprandomteiri}\label{twohoprandom:teiri}
Given the upper bound $N$ and $\Delta$, $\PLRU$ is a self-stabilizing two-hop coloring protocol with randomized transitions, and the convergence time is $O(mn)$ steps both in expectation and with high probability.
\end{restatable}

\begin{restatable}{theorem}{twohopteiri}  
Given the upper bound $N$ and $\Delta$, $\PDLRU$ is a self-stabilizing two-hop coloring protocol with \deterministic~transitions, and converges to safe configurations within $O(m(n+\Delta\log{N}))$ steps both in expectation and with high probability.  
\end{restatable}

\section{Loosely-Stabilizing Leader Election}  
In this section, we propose a loosely-stabilizing leader election protocol $\PBC$.  
$\PBC$ uses a self-stabilizing two-hop coloring protocol. Thus, if it uses $\PLRU$, $\PBC$ is with randomized transitions. If it uses $\PDLRU$, $\PBC$ is with \deterministic~transitions.  
In both cases, $\PBC$ holds a unique leader with $\Omega(Ne^{2N})$ expected steps and uses $O(\Delta\log{N})$ bits of memory.
Note that $\PBC$ (with randomized transitions) always generates random numbers deterministically like in $\PDLRU$ outside of two-hop coloring 
since it does not affect a whole complexity.

The basic strategy of leader election is as follows:  
i) All agents become followers.  
ii) Some candidates of a leader emerge, and the number of candidates becomes $1$ with high probability. 
iii) If there are multiple leaders, return to i).
Each agent $a$ has $6$ variables and $4$ timers:
$a.\LF \in \{\Baby,\LL_0,\LL_1,\FF\}$, 
$a.\type \in \{1,...,2^{\lceil\log{N}\rceil+1}-1\}$, 
$a.\iid \in \{1,...,2^{\lceil\log{N^2}\rceil+1}-1\}$,  
$a.\clr$,  
$a.\pcol$,  
$a.\rc$ $\in \{0,1\}$,      
$a.\timer_\LF $ $ \in [0,2\tbc]$,  
$a.\timer_\KL \in [0,\tbc]$,  
$a.\timer_\mathrm{V} \in [0,2\tbc]$, and  
$a.\timer_\mathrm{E} \in [0,2\tbc]$.  
Here, $\tbc$ is a sufficiently large value for $\PBC$ to work correctly. 
A status of an agent $a$ is represented by $a.\LF$ where $\Baby$ (leader candidate), $\LL_0$ (leader mode), and $\LL_1$ (duplication check mode) represent leaders and $\FF$ represents a follower.
The variable $a.\type$ is used for detecting multiple leaders.
The variable $a.\iid$ represents the identifier of leaders.

In $\PBC$, agents mainly use the broadcast (also called the epidemic and the propagation) to inform others of something.  
In a broadcast mechanism, information from one agent is repeatedly copied (with modification if needed) to agents when two agents interact.
In order to detect the end of operations (including broadcasts) for all agents, the agents uses timers. 
The timers decrease by Larger Time Propagation and Same Speed Timer. 
Using Larger Time Propagation and Same Speed Timer, all timer values decrease gradually, almost synchronously. 
Larger Time Propagation means that when an agent $u$ interacts with an agent $v$, $u.\timer$ (resp. $v.\timer$) is set to $\max(u.\timer, v.\timer-1)$ (resp. $\max(v.\timer, u.\timer-1)$). 
A variable $a.\rc$ is used to implement Same Speed Timer and represents whether $a$ can decrease its own timers or not in the current interaction.
Same Speed Timer decrease a timer value by $1$ when an agent interacts the same agent continuously.
For Same Speed Timer, after reaching $\Scol$, the agents use the colors to determine whether the current partner is the last partner or not.
A read-only variable $a.\clr$ represents a color determined in the two-hop coloring.  
A variable $a.\pcol$ represents an agent' color that $a$ interacted previously.
The domains of $a.\clr$ and $a.\pcol$ depend on a self-stabilizing two-hop coloring protocol.  
If we use $\PLRU$, $a.\clr,a.\pcol \in \{1,\dots,8N^3\Delta^2\}$.  
If we use $\PDLRU$, $a.\clr,a.\pcol \in \{0,\dots,2^{\lceil\log{8N^3\Delta^2}\rceil}-1\}$.

The overall flow of $\PBC$ is shown in Figure~\ref{figure:PBR}.

\subsection{Outline of \texorpdfstring{$\PBC$}{P\_{BC}}}

We show the outline of $\PBC$.
We call a configuration satisfying certain conditions a phase.
$\PBC$ mainly has 3 phases: i) Global Reset, ii) Leader Generation, iii) Leader Detection.
We first explain the overview of the three phases with the roles of $4$ timers using an example flow shown in Fig.~\ref{figure:PBR}, where
the height of the timer represents the relative magnitude of its value.

i) Global Reset is the phase that resets all agents when some inconsistencies are detected.
(In Fig.~\ref{figure:PBR}(a), multiple leaders are detected.)
A configuration is in the Global Reset if there exists agents whose $\timer_\KL>0$.
When some inconsistency is detected, a Global Reset phase is started and the kill virus is created.  
The kill virus makes an agent a follower, sets the agent's identifier to $1$ (a tentative value before generating id), and also erases the search virus.  
The presence of kill virus is represented as a positive value of $\timer_\KL$, which serves as timer to live (TTL).
When the Global Reset phase begins, some agents' $\timer_\KL$ are set to the maximum value and spread to all agents
with decreasing the values as shown in Fig.~\ref{figure:PBR}(b). 
Eventually, $\timer_\KL$ will become $0$, and the Global Reset phase will be finished.
While $\timer_\KL>0$, $\timer_\LF$ takes its maximum value and $\timer_\mathrm{V}$ takes 0.

ii) Leader Generation is the phase where agents generate leaders.
In $\PBC$, a leader keeps a values of $\timer_\LF$ to its maximum value, while followers propagate the value with Larger Time Propagation and Same Speed Timer.
If there is no leader
(Fig.~\ref{figure:PBR}(c)),
values of $\timer_\LF$ 
for some followers eventually become $0$. 
Then these followers become candidates ($\Baby$)
(Fig.~\ref{figure:PBR}(d)).
Each candidate for leaders generates a random number as an identifier using interactions.  
The way to generate random numbers is described in Section~\ref{sec:twohopmain}.
The candidates broadcast their $\iid$s to all agents, 
and become a follower (\(\FF\)) if it encounters a larger $\iid$ value.
While generating an identifier ($\iid$), $\timer_\mathrm{LF}$ takes its maximum value, and it gradually decreases after generating $\iid$.
When $\timer_\mathrm{LF}$ of a candidate becomes $0$, it becomes a leader (\(\LL_0\)).
When there are no candidates for leaders, Leader Generation is finished.

iii) Leader Detection is the phase where leaders determine whether there are multiple leaders or not.
The leaders generate search viruses periodically using $\timer_\mathrm{E}$.
When $\timer_\mathrm{E}$ becomes $0$, the agent start generating a $\type$ of search virus.
If there are multiple leaders, leaders generate search viruses almost simultaneously thanks to Larger Time Propagation.  
When leaders generate search viruses, they generate random numbers using interactions to determine the type of search virus.  
The type of search virus with its TTL $\timer_\mathrm{V}$ spreads to all agents.  
When two different types of search viruses meet, agents create the kill virus and move to the Global Reset phase.
(Fig.~\ref{figure:PBR}(a)).
While, if there is a unique leader, one search virus is periodically generated and expired (Fig.~\ref{figure:PBR}(e)).

Phases circulates Global Reset, Leader Generation, and Leader Detection in this order.
If there exists a unique leader , the configuration stays in Leader Detection phase with high probability.  
Otherwise, the phase moves to Global Reset phase.
Timers $\timer_\KL$ and $\timer_\LF$ are used for Global Reset and Leader Generation phases to have enough steps.

\begin{figure}[tb]
    \centering
    \includegraphics[width=1.0\linewidth]{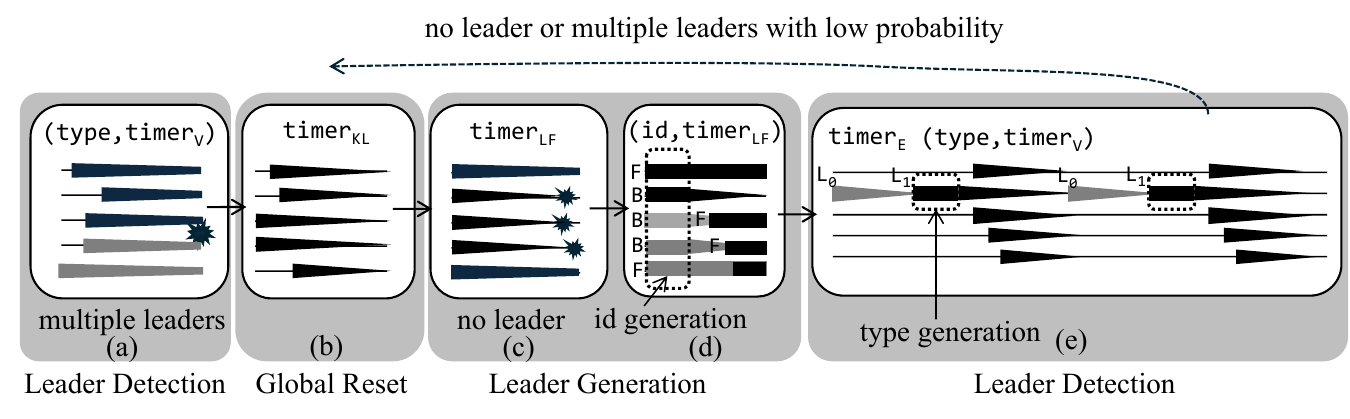}
    \caption{An example flow of $\PBC$.}
    \label{figure:PBR}
\end{figure}

\begin{algorithm}[!htb]
    \caption{Loosely-Stabilizing Leader Election Protocol $\PBC$ (1/3)}
    \label{Protocol:LSLE1}
    \SetKwInOut{OutFunc}{Output function $O$}
    \SetKwProg{Fn}{function}{:}{}
    \setcounter{AlgoLine}{0}
    \SetKwFunction{LTP}{LARGER\_TIME\_PROPAGATE}
    \SetKwFunction{REPEATCHECK}{REPEAT\_CHECK}
    \SetKwFunction{COUNTDONW}{COUNT\_DOWN}
    \SetKwFunction{LeaderGeneration}{GenerateLeader}
    \SetKwFunction{Detection}{Detect}
    \SetKwFunction{Reset}{Reset}
    \OutFunc{An agent $a$ outputs $\LL$ if $a.\LF\in \{\Baby,\LL_0,\LL_1\}$, otherwise, $\FF$.}
    \nonl\when{an initiator $a_0$ interacts with a responder $a_1$}{
        Execute self stabilizing two-hop coloring protocol\;
        $\REPEATCHECK()$\;
        \lIf{$\exists i\in\{0,1\}:a_i.\timer_\KL>0$}{
            $\Reset()$
        }
        $\LTP(\KL)$\;
        $\COUNTDONW(\KL)$\;
        $\LTP(\mathrm{E})$\;
        $\COUNTDONW(\mathrm{E})$\;
        \If{$a_0.\LF,a_1.\LF\in \{\LL_0,\LL_1,\FF\}$}{
            $\LTP(\LF)$\;
            $\COUNTDONW(\LF)$\;
        }
        \Forall{$ i\in \{0,1\}$}{
            \lIf{$a_i.\LF\in \{\LL_0,\LL_1\}$}{
                $a_i.\timer_\LF\gets \tbc$
            }
        }
        $\LeaderGeneration()$\;
        $\Detection()$\;
    }
    \nonl\Fn{\LTP{$x$}}{
        \If{$\exists i\in\{0,1\}:a_i.\timer_{x}<a_{1-i}.\timer_{x}$}{
            $a_i.\timer_{x}\gets a_{1-i}.\timer_{x}-1$\;
        }
    }
    \nonl\Fn{\REPEATCHECK{}}{
        \Forall{$i\in \{0,1\}$}{
            \lIf{$a_i.\pcol = a_{1-i}.\clr$}{$a_i.\rc\gets 1$}
            \lElse{$a_i.\rc\gets 0$}
            $a_i.\pcol\gets a_{1-i}.\clr$\;
        }
    }
    \nonl\Fn{\COUNTDONW{$x$}}{
        \Forall{$i\in\{0,1\}$}{
            \lIf{$a_i.\rc=1$}{$a_i.\timer_{x}\gets \max(0, a_i.\timer_{x}-1)$}
        }
    }
    \nonl\Fn{\Reset{}}{
        \Forall{$i\in \{0,1\}$}{
            $(a_i.\LF, a_i.\iid ,a_i.\timer_\mathrm{V})\gets (\FF, 1, 0)$\;
        }
    }
\end{algorithm}

\begin{algorithm}[!htb]
    \caption{Loosely-Stabilizing Leader Election Protocol $\PBC 
$ (2/3)}
    \label{Protocol:LSLE2}
    \nonl\Fn{\LeaderGeneration{}}{
        \If{$\exists i\in\{0,1\}:a_i.\timer_\KL>0$}{
            $a_0.\timer_\LF\gets a_1.\timer_\LF\gets \tbc$ \quad\quad\quad\quad// prevent starting GenerateLeader\;
        }
        \If{$\exists i\in \{0,1\}:a_i.\LF = \FF \wedge a_i.\timer_\LF = 0$}{
            \lIf{$a_i.\iid \ne 1$}{
                $a_0.\timer_\KL\gets a_1.\timer_\KL\gets \tbc$\quad// $\iid$ hasn't been reset
            }\lElse{
                $(a_i.\LF, a_i.\timer_\LF) \gets (\Baby, 2\tbc)$\quad// a new candidate is created
            }
        }
        \If{$\forall i\in\{0,1\}:a_i.\LF=\Baby \wedge a_i.\iid<2^{\lceil\log{N^2}\rceil}$}{
            $(a_1.\LF, a_1.\iid, a_1.\timer_\LF) \gets (\FF, 1, \tbc)$\;
        }
        \If{$\exists i\in \{0,1\}:a_i.\LF = \Baby \wedge a_i.\iid < 2^{\lceil\log{N^2}\rceil}$}{
                $(a_i.\iid, a_i.\timer_\LF)\gets (2a_i.\iid + i, 2\tbc)$\;
        }
        \If{$a_0.\LF\in \{\FF,\Baby\}\wedge a_1.\LF\in\{\FF,\Baby\} \wedge (\exists i\in\{0,1\}:((a_i.\LF=\Baby \wedge a_i.\iid > 2^{\lceil\log{N^2}\rceil})\vee a_i.\LF = \FF)\wedge a_i.\iid > a_{1-i}.\iid)$}{
                $(a_{1-i}.\LF,a_{1-i},\iid)\gets(a_i.\LF,a_i.\iid)$\;
        }
        \If{$\exists i\in\{0,1\}:a_i.\LF=\FF \wedge a_{1-i}.\LF=\Baby$}{
            $a_i.\timer_\LF\gets \tbc-1$ \quad\quad\quad\quad// consider $a_{1-i}$ as a leader\;
        }
        \Forall{$i\in \{0,1\}$}{
            \lIf{$a_i.\LF =\Baby \wedge a_i.\rc=1$}{
                $a_i.\timer_\LF \gets \max(0, a_i.\timer_\LF-1)$
            }
            \lIf{$a_i.\LF = \Baby \wedge a_i.\timer_\LF = 0$}{
                $(a_i.\LF, a_i.\timer_\LF) \gets (\LL_0, \tbc)$
            }
        }
    }
\end{algorithm}

\begin{algorithm}[!htb]
    \caption{Loosely-Stabilizing Leader Election Protocol $\PBC$ 
 (3/3)}
    \label{Protocol:LSLE3}
    \nonl\Fn{\Detection{}}{
        \If{$\exists i\in\{0,1\}:a_i.\LF=\Baby$}{
            $a_0.\timer_\mathrm{E}\gets a_1.\timer_\mathrm{E}\gets 2\tbc$ \quad\quad// prevent starting Detect\;
        }
        \Forall{$i\in \{0,1\}$}{
            \If{$a_i.\LF \in \{\LL_0,\LL_1\} \wedge a_i.\timer_\mathrm{E} = 0$}{
                $(a_i.\LF, a_i.\type, a_i.\timer_\mathrm{E})\gets (\LL_1, 1, 2\tbc)$ // start the type generation\;
            }
            \If{$a_i.\LF = \LL_1 \wedge a_i.\type < 2^{\lceil\log{N}\rceil}$}{
                $(a_i.\type, a_i.\timer_\mathrm{E}) \gets (2 a_i.\type + i, 2\tbc)$\;
                \If{$a_i.\type \ge 2^{\lceil\log{N}\rceil}$}{
                    $a_i.\timer_\mathrm{V} \gets 2\tbc$\;
                }
            }
        }
        \uIf{$\exists i\in\{0,1\}:a_i.\LF=\FF  \wedge a_{1-i}.\LF\in\{\LL_0,\LL_1\}$}{
            \uIf{$a_i.\timer_\mathrm{V}>0 \wedge (a_{1-i}.\LF=\LL_0 \vee a_{1-i}.\type<2^{\lceil\log{N}\rceil} \vee$\\ \nonl$ a_i.\type \ne a_{1-i}.\type)$}{
                $a_0.\timer_\KL\gets a_1.\timer_\KL \gets \tbc$ \quad// different types are detected\;
            }\ElseIf{$a_i.\timer_\mathrm{V}=0 \wedge a_{1-i}.\LF=\LL_1 \wedge a_{1-i}.\timer_\mathrm{V}>0$}{
                $(a_i.\type, a_i.\timer_\mathrm{V}) \gets (a_{1-i}.\type, a_{1-i}.\timer_\mathrm{V}-1)$\;
            }
        }
        \uElseIf{$a_0.\LF = \FF \wedge a_1.\LF = \FF$}{
            \uIf{$a_0.\timer_\mathrm{V} > 0 \wedge a_1.\timer_\mathrm{V} > 0 \wedge a_0.\type \ne a_1.\type$}{
                $a_0.\timer_\KL \gets a_1.\timer_\KL \gets \tbc$\quad // different types are detected\;
            }
            \ElseIf{$\exists i\in\{0,1\}:a_i.\timer_\mathrm{V} = 0 \wedge a_{1-i}.\timer_\mathrm{V}>0$}{
                $(a_i.\type, a_i.\timer_\mathrm{V})\gets (a_{1-i}.\type, a_{1-i}.\timer_\mathrm{V}-1)$\;
            }
        }
        \ElseIf{$a_0.\LF\in\{\LL_0,\LL_1\} \wedge a_1.\LF \in \{\LL_0, \LL_1\}$}{
            $a_0.\timer_\KL\gets a_1.\timer_\KL\gets \tbc$\quad // multiple leaders are detected\;
        }
        $\LTP(\mathrm{V})$\;
        $\COUNTDONW(\mathrm{V})$\;
        \Forall{$i\in \{0,1\}$}{
            \lIf{$a_i.\timer_\mathrm{V}>0$}{
               $a_i.\timer_\mathrm{E}\gets 2\tbc$// prevent restarting Detect 
            }
            \lIf{$a_i.\LF = \LL_1 \wedge a_i.\timer_\mathrm{E}<\tbc/2$}{
                $a_i.\LF \gets \LL_0$
            }
        }
    }
\end{algorithm}

\subsection{Details of \texorpdfstring{$\PBC$}{P\_{BC}}}

$\PBC$ is given by Algorithm~\ref{Protocol:LSLE1}, Algorithm~\ref{Protocol:LSLE2}, and Algorithm~\ref{Protocol:LSLE3}.
We explain the details of $\PBC$.  
$\PBC$ has five parts:  
i) Two-hop coloring, ii) Timer Count Down, iii) Reset, iv) Leader Generation, v) Leader Detection.  
The relationship between the three phases and five parts is as follows. The Global Reset phase corresponds to the Reset part, the Leader Generation phase to the Generate part, and the Leader Detection phase to the Detect part. The Timer Count Down part operates throughout all phases, while the Two-Hop Coloring part is completed before the phases begin.
Throughout this explanation, we consider when an initiator $a_0$ interacts with a responder $a_1$.

i) In line 1, the agents execute the two-hop coloring protocol $\PLRU$ or $\PDLRU$.

ii) Timer Count Down (lines 2, 4--12, and 15--22) aims to increase or decrease timers.  
First, the interacting agents determine whether the current partner is the same as the last partner in REPEAT\_CHECK (lines 2, and 17--20) to implement Same Speed Timer.  
After that, each agent saves the current partner's color to its own $\pcol$.  
Then, agents decrease timers if $a_i.\rc=1$ holds.  
$\timer_\KL$ (lines 4--5) and $\timer_\mathrm{E}$ (lines 6--7) are handled by Larger Time Propagation and Count Down.  
That is, for $i\in \{0,1\}$, $a_i.\timer_\KL$ is set to $\max(a_i.\timer_\KL, a_{1-i}.\timer_\KL-1)$, and if $a_i.\rc=1$ holds, $a_i.\timer_\KL$ is decreased by 1 (resp. $\timer_\mathrm{E}$).  
If both interacting agents are not candidates ($\Baby$), they increase or decrease $\timer_\LF$ like $\timer_\KL$ (lines 8--10).  
For $i\in\{0,1\}$, if $a_i$ is a leader ($\LL_0$ or $\LL_1$), $a_i.\timer_\LF$ is set to $\tbc$ (lines 11--12).  
If $a_0$ or $a_1$ is a candidate ($\Baby$), $\timer_\LF$ increases or decreases in Leader Generation.  
$\timer_\mathrm{V}$ increases or decreases in Leader Detection.

iii) Reset (lines 3, and 23--24) aims to reset the population when some inconsistency is detected.  
For $i\in \{0,1\}$, if $a_i.\timer_\KL>0$ holds, $a_i$ sets $(a_i.\LF, a_i.\iid, a_i.\timer_\mathrm{V})$ to $(\FF, 1, 0)$.  
In other words, $a_i$ becomes a follower, its identifier becomes $1$, and $a_i$ erases the search virus.
This Reset is happened when and only when there are multiple leaders and candidates' $\iid$ have not been reset to $1$.
Specifically, when different search virus meets (lines 60--61, 65--66, 69--70), and when candidates' $\iid$ have not been reset (lines 27--28).

iv) Generate Leader (lines 13, and 25--40) aims to generate a new leader when there are no leaders.  
For $i\in \{0,1\}$, if $a_i.\timer_\KL>0$ holds, each agent sets $a_i.\LF$ to $\tbc$ to prevent starting Leader Generation during the Global Reset phase (lines 25--26).  
Firstly, for $i\in \{0,1\}$, if $a_i.\LF$ becomes $0$, $a_i$ determines that there are no leaders, and becomes a candidate for leaders ($\Baby$) and sets $a_i.\timer_\LF$ to $2\tbc$ (lines 27--29).  
At this time, if $a_i.\iid$ is not $1$ (\ie, it has not been reset), $a_i$ sets $a_i.\timer_\KL$ to $\tbc$ and moves to the Global Reset phase (line 29).  
Secondly, candidates ($\Baby$) generate random numbers as their own identifiers ($\iid$).
For each interaction, if the candidate $u$ is an initiator, $u.\iid$ is updated to $2u.\iid$; otherwise, $u.\iid$ is updated to $2u.\iid+1$ until $u.\iid$ becomes no less than $2^{\lceil\log{N^2}\rceil}$ (lines 32--33).  
For the independence of random numbers, if both agents are candidates and generating random numbers, the responder becomes a follower ($\FF$) by resetting $\iid$ to $1$ (lines 30--31).  
Thirdly, if a candidate's $\iid$ becomes no less than $2^{\lceil\log{N^2}\rceil}$, the candidate starts to broadcast its own $\iid$ to other agents (lines 34--35).  
This broadcast allows all agents to know the maximum $\iid$ of candidates.  
For $i\in \{0,1\}$, if $a_i$ is a candidate and $a_i.\iid < a_{1-i}.\iid$ holds, $a_i$ becomes a follower ($\FF$) and sets $a_i.\iid$ to $a_{1-i}.\iid$.  
For $i\in \{0,1\}$, if $a_i$ is a follower and $a_i.\iid > a_{1-i}.\iid$ holds, $a_{1-i}$ sets $a_{1-i}.\iid$ to $a_i.\iid$.  
To avoid generating new candidates when there are candidates in the population, for $i\in \{0,1\}$, if $a_i$ is a follower and $a_{1-i}$ is a candidate, $a_i$ sets $a_i.\timer_\LF$ to $\tbc-1$ (lines 36--37).
That is, we consider $a_{1-i}$ has $\timer_\LF=\tbc$ virtually.
Eventually, all agents' $\iid$s become the same, and most candidates become followers (and some candidates remain).  
The candidates measure until all candidates finish generating the $\iid$ and $\iid$s are broadcast for a sufficiently long time using $\timer_\LF$.  
A candidate decreases $\timer_\LF$ by 1 if the agent's $\rc$ is $1$ (lines 39).  
Finally, when a candidate's $\timer_\LF$ becomes $0$, the candidate becomes a new leader ($\LL_0$) and sets $\timer_\LF$ to $\tbc$ (lines 40).  
The range of generated identifiers ($\iid$) is $[2^{\lceil\log{N^2}\rceil},2^{\lceil\log{N^2}\rceil+1})$, so there exists a unique leader with high probability.

v) Detect (lines 41--66) aims to determine whether there are multiple leaders or not.  
Leaders generate search viruses every time their $\timer_\mathrm{E}$ becomes $0$.  
If $a_0$ or $a_1$ is a candidate ($\Baby$), agents set their $\timer_\mathrm{E}$ to $2\tbc$ to prevent generating a search virus (lines 41--42).
Firstly, for $i\in \{0,1\}$, if $a_i$ is a leader ($\LL_0$ or $\LL_1$) whose $\timer_\mathrm{E}$ becomes $0$, $a_i$ becomes $\LL_1$ and starts generating random numbers to get the type of search virus (lines 44--45).  
At the beginning of generating random numbers, $a_i$ sets $a_i.\type$ to $1$.  
The way of generating random numbers is the same as $\iid$ generation.  
While generating random numbers, a leader sets own $\timer_\mathrm{E}$ to $2\tbc$ to inform that there exist agents generating random numbers (lines 46--47).
When a leader finished generating random numbers, the leader sets own $\timer_\mathrm{V}$ to $2\tbc$ (lines 48--49).
Secondly, agents detect multiple leaders if there are multiple leaders.
$\LL_1$ broadcasts the generated search virus to all agents via some agents until $\timer_\mathrm{V}$ becomes $0$ (lines 50--59). 
If a follower having search virus and a leader not having search virus interact except the cases their types are same, they set $\timer_\KL$ to $\tbc$ and the phase moves to Global Reset (lines 51--52).
If a follower not having search virus and a leader having search virus interact, the follower set own $\timer_\mathrm{V}$ to the leader's $\timer_\mathrm{V}-1$ and set own $\type$ to the leader's $\type$ (lines 53--54).
If $a_0$ and $a_1$ are followers and they have different types of search viruses, they set $\timer_\KL$ to $\tbc$ and move to Global Reset phase (lines 56--57).  
If $a_0$ and $a_1$ are followers and there exists $a_i,a_{1-i}$ agents satisfying $a_i.\timer_\mathrm{V}=0$ and $a_{1-i}.\timer_\mathrm{V}>0$ for $i\in\{0,1\}$, $a_i.\timer_\mathrm{V}$ is set to $a_{1-i}.\timer_\mathrm{V}-1$ and $a_i.\type$ is set to $a_{1-i}.\type$ (lines 58--59).  
If $a_0$ and $a_1$ are leaders, they set $\timer_\KL$ to $\tbc$ and move to Global Reset phase (lines 60--61).  
Finally, both agents' $\timer_\mathrm{V}$ run Larger Time Propagation and decrease by $1$ if $a_i.\rc=1$ holds (lines 62--63).  
For $i\in \{0,1\}$, if $a_i.\timer_\mathrm{V}>0$ holds, $a_i.\timer_\mathrm{E}$ is set to $2\tbc$ to prevent generating a new search virus when there is a search virus in the population (line 65).  
For $i\in \{0,1\}$, if $a_i$ is a leader and $a_i.\timer_\mathrm{E}$ becomes less than $\tbc/2$, $a_i$ becomes $\LL_0$ (line 66).  
The range of generating random numbers of types is $[2^{\lceil\log{N}\rceil},2^{\lceil\log{N}\rceil+1})$, so when there are multiple leaders in the population, the types generated by leaders are not the same with high probability.

The proof of Lemma~\ref{leaderelection:independent} in Appendix.

\begin{restatable}{lemma}{LEindependent}\label{leaderelection:independent}
For any execution, all candidates' $\iid$ no less than $2^{\lceil\log{N^2}\rceil}$are independent and uniform if they are started to be generated during the execution. All leaders' $\type$ no less than $2^{\lceil\log{N}\rceil}$ are independent and uniform if they are generated from the beginning of this execution.  
\end{restatable}

\newcommand{\assumeN}{$\tau \ge \max(2d, \lceil\log{N}\rceil/2, 15+3\log{n})$}
\subsection{Analysis}  
In this subsection, we analyze the expected convergence time and the expected holding time of $\PBC$.
We assume \assumeN, and $\tbc=16\tau$.  
We will prove the following equations under these assumptions:  

$\max_{C\in \Scol}\ECT_\PBC(C,\SLE)=O(m\tau\log{n})$.

$\min_{C\in \SLE}\EHT_\PBC(C,LE)=\Omega(\tau e^\tau)$.

Here, $\Scol$ and $\SLE$ are the sets of configurations described later.

We define the sets of configurations to prove the above equations:

$\Scol$ is the safe configurations of the self-stabilizing two-hop coloring. 

$\SDcol\subset \Scol$ is the set of configurations where each agent's $\pcol$ is the same as the last interacted agent's color.

$\KLzero = \{C \in \SDcol \mid \forall v \in V:C(v).\timer_\KL=0\}$.

$\Bno = \{C \in \SDcol \mid \forall v \in V:C(v).\LF \ne \Baby\}$.

$\Lone = \{C \in \SDcol \mid |\{v\in V\mid C(v).\LF \in \{\LL_0,\LL_1\}\}|=1\}$.

$\LFqua = \{C \in \SDcol \mid \forall v \in V:C(v).\LF \ne \Baby \Rightarrow C(v).\timer_\LF \ge \tbc/2\}$.

$\Lvone = \{C \in \SDcol \mid \exists v \in V:C(v).\LF = \LL_1\}$.

$\Vclean = \{C \in \SDcol \mid \forall v \in V:C(v).\timer_\mathrm{V} = 0\}$.

$\Vmake = \{C \in \SDcol \mid \forall v \in V:(C(v).\LF = \LL_1 \Rightarrow C(v).\type < 2^{\lceil\log{N}\rceil}) \wedge (C(v).\LF\ne \LL_1 \Rightarrow C(v).\timer_\mathrm{V}=0)\}$.

$\Vonly = \{C \in \SDcol \mid \forall v,\forall u \in V:C(v).\timer_\mathrm{V} > 0 \wedge C(u).\timer_\mathrm{V} > 0 \Rightarrow C(v).\type = C(u).\type\}\cap \{C \in \SDcol \mid \forall v \in V:C(v).\LF = \LL_1 \Rightarrow C(v).\type \ge 2^{\lceil \log{N} \rceil}$.

$\Ehalf = \{C \in \SDcol \mid \forall v \in V:C(v).\LF \in \{\LL_0, \LL_1\} \Rightarrow C(v).\timer_\mathrm{E} \ge \tbc\}$.

$\SLE =\Bno \cap \Lone \cap \LFqua \cap \KLzero \cap (\Vclean \cup (\Lvone \cap (\Vmake \cup \Vonly)\cap \Ehalf))$.

\subsubsection{Expected Holding Time}

\begin{lemma}\label{leaderelection:holding:katei}
Let $C_0 \in \SLE$ and $\Xi_\PBC(C_0)=C_0,C_1,\dots$.  
If $\Pr(\forall i\in[0,2m\tau]:C_i\in LE \wedge C_{2m\tau}\in \SLE) = 1 - O(ne^{-\tau})$ holds, then $\min_{C\in \SLE}\EHT_\PBC(C,LE)=\Omega(\tau e^\tau)$ holds.
\end{lemma}

\begin{proof}  
Let $A=\min_{C_0\in \SLE}\EHT_\PBC(C_0,LE)$.
We assume that $C_0,\dots,C_{2m\tau}\in LE \wedge C_{2m\tau}\in \SLE$ holds with probability at least $p=1-O(ne^{-\tau})$.
Then, We have $A\ge p(2m\tau+A)$.
Solving this inequality gives $A\ge 2m\tau/(1-p)=\Omega(\tau e^{\tau})$.
\end{proof}

We say that an agent $u$ encounters a counting interaction when $u$ interacts with an agent $v$ such that $u.\clr = v.\pcol$ holds.
The proofs of Lemma~\ref{leaderelection:holding:count} and  Lemma~\ref{leaderelection:holding:spreads} are in Appendix.

\begin{restatable}{lemma}{LEholdcount}\label{leaderelection:holding:count}
Let $C_0\in\SDcol$ and $\Xi_\PBC(C_0)=C_0,C_1,\dots$.  
The probability that every agent encounters less than $\tbc/2$ counting interactions while $\Gamma_0,\dots,\Gamma_{2m\tau-1}$ is at least $1-ne^{-\tau}$. 
\end{restatable}

\begin{restatable}{lemma}{LEholdspreads}\label{leaderelection:holding:spreads}
Let $C_0\in \SDcol$ and $\Xi_\PBC(C_0)=C_0,C_1,\dots$.  
For any $x\in \{\LF,\KL,\mathrm{E},\mathrm{V}\}$, and for any $y\ge \tbc/2$ such that $y$ is no more than the maximum value of the domain of $\timer_x$, when $\exists v\in V:C_0(v).\timer_x\ge y$ holds, the probability that $\forall u\in V:C_{2m\tau}(u).\timer_x> y-\tbc/2$ holds is at least $1-2ne^{-\tau}$.
\end{restatable}

\begin{lemma}\label{leaderelection:holding:LFlambda}
Let $C_0\in\SDcol$ and $\Xi_\PBC(C_0)=C_0,C_1,\dots$.
For any integer $\lambda>0$, and any integer $x$ satisfying $\tbc \le x \le 2\tbc$, if $\forall v\in V:C_0(v).\timer_\LF\ge x \wedge \forall i\in[0,\lambda-1],\exists v\in V:C_{2mi\tau}(v).\timer_\LF\ge x$ holds, then $\Pr(\forall j\in[0,2m\lambda \tau],\forall v\in V:C_j(v).\timer_\LF>x-\tbc \wedge C_{2m\lambda\tau}(v).\timer_\LF\ge x-\tbc/2)\ge 1-3\lambda ne^{-\tau}$ holds.    
\end{lemma}

\begin{proof}  
Since there exists an agent $u$ satisfying $u.\timer_\LF\ge x$ in $C_0$, the probability that every agent's $\timer_\LF\ge x-\tbc/2$ holds in $C_{2m\tau}$ is at least $1-2ne^{-\tau}$ from Lemma~\ref{leaderelection:holding:spreads}.  
Since there is every agent $u$ satisfying $u.\timer_\LF\ge x-\tbc/2$ in $C_0$, $\Pr(\forall j\in[0,2m\tau],\forall v\in V:C_{j}(v).\timer_\LF>x-\tbc)\ge 1-ne^{-\tau}$ holds from Lemma~\ref{leaderelection:holding:count}.  
Thus, $\Pr(\forall j\in[0,2m\tau],\forall v\in V:C_{j}(v).\timer_\LF>x-\tbc \wedge C_{2m\tau}(v).\timer_\LF\ge x-\tbc/2)\ge 1-3ne^{-\tau}$ holds by the union bound.  
Repeating this $\lambda$ times, we get $\Pr(\forall j\in[0,2m\lambda\tau],\forall v\in V:C_{j}(v)>x-\tbc \wedge C_{2m\lambda\tau}(v).\timer_\LF\ge x-\tbc/2)\ge 1-3\lambda ne^{-\tau}$ by the union bound.  
\end{proof}

We can prove Lemma~\ref{leaderelection:holding:Elambda} by the same way of Lemma~\ref{leaderelection:holding:LFlambda}, and  Lemma~\ref{leaderelection:holding:LF} by assigning $\lambda=1$ to Lemma~\ref{leaderelection:holding:LFlambda}. 
Lemma~\ref{leaderelection:holding:v1clean}, ~\ref{leaderelection:holding:v1makeehalf}, and~\ref{leaderelection:holding:v1onlyehalf} analyzes the probability that configuration keep some condition for some interval. (See Appendix for the proofs).

\begin{lemma}\label{leaderelection:holding:Elambda}
Let $C_0\in\SDcol$ and $\Xi_\PBC(C_0)=C_0,C_1,\dots$.
For any integer $\lambda>0$, and any integer $x$ satisfying $\tbc \le x \le 2\tbc$, if $\forall v\in V:C_0(v).\timer_\mathrm{E}\ge x \wedge \forall i\in[0,\lambda-1],\exists v\in V:C_{2mi\tau}(v).\timer_\mathrm{E}\ge x$ holds, then $\Pr(\forall i\in[0,2m\lambda \tau],\forall v\in V:C_i(v).\timer_\mathrm{E}>x-\tbc \wedge C_{2m\lambda\tau}(v).\timer_\mathrm{E}\ge x-\tbc/2)\ge 1-3\lambda ne^{-\tau}$ holds. 
\end{lemma}

\begin{lemma}\label{leaderelection:holding:LF}
Let $C_0\in \SLE$ and $\Xi_\PBC(C_0)=C_0,C_1,\dots$.  
$\Pr(\forall i\in[0,2m\tau]:C_i\in \Bno \wedge C_{2m\tau}\in \LFqua)\ge 1-3ne^{-\tau}$ holds. 
\end{lemma}

\begin{restatable}{lemma}{LEholdvoneclean}\label{leaderelection:holding:v1clean}  
Let $C_0\in \SLE \cap \Vclean$ and $\Xi_\PBC(C_0)=C_0,C_1,\dots$.  
$\Pr(\forall i\in[0,2m\tau]:C_i\in \Lone\wedge C_{2m\tau}\in \SLE \cap (\Vclean \cup \Lvone \cap (\Vmake \cup \Vonly) \cap \Ehalf)\ge 1-5ne^{-\tau}$ holds.  
\end{restatable}

\newcommand{\Lnotincreasing}{From Lemma~\ref{leaderelection:holding:LF}, the probability that the number of leaders does not increase from $C_0$ to $C_{2m\tau}$ and $C_{2m\tau}$ in $\LFqua$ is at least $1-2ne^{-\tau}$.  
Thus, $C_0,\dots, C_{2m\tau}\in \Bno$ in this case.}

\begin{restatable}{lemma}{LEholdvonemakeehalf}\label{leaderelection:holding:v1makeehalf}  
Let $C_0\in \SLE \cap \Lvone \cap \Vmake \cap \Ehalf$ and $\Xi_\PBC(C_0)=C_0,C_1,\dots$.  
$\Pr(\forall i\in[0,2m\tau]:C_i\in  \Lone \wedge C_{2m\tau}\in \SLE \cap \Lvone \cap (\Vmake \cup \Vonly)\cap \Ehalf)\ge 1-3ne^{-\tau}$ holds. 
\end{restatable}

\begin{restatable}{lemma}{LEholdvoneonlyehalf}\label{leaderelection:holding:v1onlyehalf}  
Let $C_0\in \SLE \cap \Lvone \cap \Vonly \cap \Ehalf$ and $\Xi_\PBC(C_0)=C_0,C_1,\dots$.  
$\Pr(\forall i\in[0,2m\tau]:C_i\in \Lone \wedge C_{2m\tau}\in \SLE \cap (\Vclean \cup \Lvone \cup \Vonly \cap \Ehalf)\ge 1-5ne^{-\tau}$ holds.  
\end{restatable}

\begin{lemma}\label{leaderelection:holding:hold} 
$\min_{C \in \SLE}\EHT_\PBC(C,LE)=\Omega(\tau e^{\tau})$. 
\end{lemma}

\begin{proof}  
Let $C_0\in \SLE$.  
$\Pr(C_0,\dots,C_{2m\tau}\in LE \wedge C_{2m\tau}\in \SLE)\ge 1-5ne^{-\tau}=1-O(ne^{-\tau})$ from Lemma~\ref{leaderelection:holding:v1clean}, Lemma~\ref{leaderelection:holding:v1makeehalf}, and Lemma~\ref{leaderelection:holding:v1onlyehalf}.  
Thus, this lemma follows from Lemma~\ref{leaderelection:holding:katei}.  
\end{proof}

\subsubsection{Expected Convergence Time}

We first analyze the number of interactions until all timers converge to $0$ with high probability.
Let $\lambda\tbc$ be the maximum value of the domain of timers, that is, $\lambda=1$ for  $\timer_\KL$ and $\lambda=2$ for $\timer_\LF$, $\timer_\mathrm{V}$, and $\timer_\mathrm{E}$.
\newcommand{\convtime}{2340\lambda m\tau\log{n}}

\begin{lemma}\label{leaderelection:conv:timeconv}
Let $C_0\in \SDcol$ and $\Xi_\PBC(C_0)=C_0,C_1,\dots$.  
For any $x\in \{\LF,\KL,\mathrm{V},\mathrm{E}\}$, if every agent's $\timer_x$ increases only by Larger Time Propagation (not including setting to a specific value like $\tbc$ by leaders etc.), the number of interactions until  every agent's $\timer_x$ becomes $0$ is less than $\convtime$ with probability at least $1-e^{-\lambda \tau}$.
\end{lemma}

\begin{proof}
Let $z=\max_{v\in V}(C_i(v).\timer_x)\ (i> 0)$.  
From the mechanism of Larger Time Propagation, for every agent $v$, $C_i(v).\timer_x$ does not become $z$ if $C_{i-1}(v).\timer_x < z$.  
Thus, when every agent decreases its timer by at least $1$ from $C_j,\dots,C_i\ (0\le j<i)$, $\max_{v\in V}(C_i(v).\timer_x)-\max_{v\in V}(C_j(v).\timer_x)\ge 1$ holds (\ie, the maximum value of $\timer_x$ decreases by at least 1).  
Let $X\sim \text{Bi}(2m,\delta_v/m)$ be a binomial random variable that represents the number of interactions of an agent $v$ interacts during $2m$ interactions.  
From Lemma~\ref{Arisu_upper}, 
$\Pr(X\ge \delta_v)\ge \Pr(X> \delta_v)=1-\Pr(X\le \delta_v)=1-\Pr(X\le (1-1/2)E[X])\ge 1-e^{-\delta_v/8}\ge 1-e^{-1/4}> 1/5$.  
Thus, the probability that an agent $v$ interacts no less than $\delta_v$ times during $2m$ interactions is at least $1/5$.  
Let $Y\sim \text{Bi}(\delta_v,2/\delta_v)$ be a binomial random variable that represents the number of counting interactions that an agent $v$ encounters during $\delta_v$ interactions in which an agent $v$ interacts.  
The probability that $Y=0$ is $\Pr(Y=0)=(1-2\delta_v)^{\delta_v}\le e^{-2}< 1/5$.  
Thus, the probability that an agent $v$ encounters at least one counting interaction is at least $4/5$.  
Let $E_v$ denote the number of interactions until an agent $v$ decreases its $\timer_x$ by at least $1$.  
Since $E_v\le 2m+(1-4/25)E_v$ holds, $E_v\le 13m$ holds.  
By Markov's inequality, the probability that an agent $v$ does not decrease $\timer_x$ during $2E_v$ interactions is no more than $1/2$.  
Thus, the probability that an agent $v$ does not decrease $\timer_x$ during $4\log{n}\cdot E_v$ interactions is no more than $n^{-2}$.  
By the union bound, the probability that every agent $v$ does not decrease $\timer_x$ during $4\log{n}\cdot E_v$ interactions is no more than $n^{-1}$.  
Let $A$ be an event that every agent $v$ decreases $\timer_x$ by at least 1 during $4\log{n}\cdot E_v$ interactions.  
We consider the expected number of times until $A$ succeeds $16\lambda\tau(=\lambda\tbc)$ times using geometric distributions.  
In other words, for $k\in[1,16\lambda \tau]$, let $Z_k\sim Geom(p_k)$ be the independent geometric random variable such that $p_k=1-1/n\ge 1/2$.  
Considering the sum of independent random variables $Z=\sum_{k=1}^{16\lambda \tau}Z_k$.  
From Lemma~\ref{Janson_lower}, 
$\Pr(Z\ge 45\lambda \tau)\le \Pr(Z\ge 1.4\cdot 32\lambda \tau)\le \Pr(Z\ge 1.4\cdot E[Z])\le e^{-p_k\cdot E[Z](1.4-1-\log_e{1.4})}\le e^{-16\lambda \tau (1.4-1-\log_e{1.4})}\le e^{-\lambda \tau}$.
Note that $p_k\cdot E[Z]=16\lambda\tau$ holds.
Thus, the expected number of times that $A$ succeeds $16\lambda\tau$ times is less than $45\lambda \tau$ with probability at least $1-e^{-\lambda \tau}$.  
Therefore, the number of interactions until all agents' $\timer_x$ becomes $0$ is $45\lambda \tau \cdot 4\log{n}\cdot E_v\le \convtime$ with  probability at least $1-e^{-\lambda\tau}$.  
\end{proof}

Lemma~\ref{leaderelection:conv:conv} shows a convergence time.

\begin{restatable}{lemma}{LEconvconv}\label{leaderelection:conv:conv}  
Let $C_0\in \Scol$ and $\Xi_\PBC(C_0)=C_0,C_1,\dots$.  
The number of interactions until the configuration reaches $\SLE$ is $O(m\tau\log{n})$ with probability $1-o(1)$.   
\end{restatable}

We have the following theorems from the above (See Appendix for the proofs).

\begin{restatable}{theorem}{ndteiri}  
Protocol $\PBC$ is a randomized $(O(m(n+\tau\log{n})),\Omega(\tau e^{\tau}))$-loosely-stabilizing leader election protocol for arbitrary graphs when \assumeN~ if $\PLRU$ is used for two-hop coloring.  
\end{restatable}

\begin{restatable}{theorem}{dteiri}  
Protocol $\PBC$ is a \deterministic~$(O(m(n+\Delta\log{N}+\tau\log{n})),\Omega(\tau e^{\tau}))$-loosely-stabilizing leader election protocol for arbitrary graphs when \assumeN~ if $\PDLRU$ is used for two-hop coloring.
\end{restatable}

\section{Conclusion}
New loosely-stabilizing leader election population protocols on arbitrary graphs without identifiers are proposed. One is randomized, and the other is deterministic. 
The randomized one converges within $O(mN\log{n})$ steps, while the deterministic one converges $O(mN\log{N})$ steps both in expectations and with high probability. Both protocols hold a unique leader with $\Omega(Ne^{2N})$ expected steps and utilizes $O(\Delta\log{N})$ bits of memory. The convergence time is close to the known lower bound of $\Omega(mN)$.

\bibliography{ref}

\clearpage

\appendix

\section{Details of Self-Stabilizing Two-Hop Coloring Protocols}

\subsection{Analysis of \texorpdfstring{$\PLRU$}{P\_{LRU}}}

In this subsection, we analyze the protocol $\PLRU$.

An agent $u$ has a color collision if and only if there exist agents $v, w \in V$ such that $(v,u) \in E$ and $(v,w) \in E$, and it holds that $u.\hop = w.\hop$.
An agent $u$ has a stamp collision if and only if there exist minimum indices $i, j$ within the range [1, $\Delta$] for which there exists a agents $v \in V$, satisfying $(u,v) \in E$ and $(u.\hop, u.\prev[i]) = (v.\prev[j], v.\hop)$, and $u.\stamp[i] \neq v.\stamp[j]$.

$\Ncc$ denotes the set of configurations in which all agents have no color collisions.
$\Scol$ denotes the set of configurations in which all agents have neither stamp collisions nor color collisions.

\begin{lemma}[Closure]\label{twohoprandom:closure}
Let $C_0\in \Scol$, and $\Xi_\PLRU(C_0)=C_0,C_1,\dots$.
$\forall i\in \mathbb{N} :C_i \in \Scol$ holds and each agent's color does not change from $C_0$.
\end{lemma}

\begin{proof}
For any configuration in \( \Scol \), all agents are free of stamp collisions and color collisions; therefore, no agent changes its color. Consequently, no configuration deviates from \( \Scol \).
\end{proof}

We analyze the time complexity in a manner similar to that described by Sudo~\etal\cite{SameSpeedTimer}. However, our analysis extends to both in expectations and with high probability.

We define an agent \( u \) as the intermediate agent if and only if there is at least one pair of distinct agents \( v, w \) such that \( (u,v) \in E \wedge (u,w) \in E \wedge v.\hop = w.\hop \) holds. We denote by \( \mathcal{C}_k \) the set of configurations in which exactly \( k \) agents has color collisions, where \( 2 \leq k \leq n \).
Furthermore, we denote by \( E_C \) the expected number of interactions required for any agent within a configuration \( C \) to generate colors (\ie when collision is detected). The maximum expected number of interactions across all configurations in \( \mathcal{C}_k \) is denoted by \( E_k = \max_{C \in \mathcal{C}_k} E_C \).

\begin{lemma}\label{twohoprandom:ek}
Let $C_0 \in \Call(\PLRU)$, and $\Xi_\PLRU(C_0) = C_0, C_1, \dots$.
$E_k \leq 320m(1/k + 1)$.
\end{lemma}

\begin{proof}
Let $u$ be an intermediate agent, let $v$ be the interaction partner when $u$ recognizes a color collision, and let $w$ be the agent whose color is the same as $v$'s.
We consider the situation in which the color collision is detected when:
i) $u$ interacts with $v$,
ii) $u$ interacts with $w$, and
iii) $u$ interacts with $v$ again.
The expected number of interactions until i) occurs is no more than $m/(2k)$, as there are $k$ agents having color collisions.
We analyze the probability that ii) and iii) occur within $m/2$ interactions after i) occurred.
We assume that $u$, $v$, and $w$'s colors do not change during the $m/2$ interactions.
Considering the case where $v$ only interacts fewer than $\Delta$ times other than with $u$, and $u$ only interacts fewer than $\Delta-1$ times other than with $v$ and $w$ after i) occurs and before iii) occurs within $m/2$ interactions.
When $\delta_v > 2$, let $X_v$ be a binomial random variable that represents the number of times that $v$ interacts with others excluding $u$ within $m/2$ interactions.
Thus, $X_v \sim \text{Bi}(m/2, (\delta_v-2)/m)$.
Note that $\delta_v$ is an even number. Using Lemma~\ref{Arisu_lower},
$\Pr(X_v \leq \Delta-1) \geq \Pr(X_v \leq \delta_v-2) = 1 - \Pr(X_v \geq \delta_v-2) = 1 - \Pr(X_v \geq (1+1)E[X_v]) \geq 1 - e^{-(\delta_v-2)/6} \geq 1 - e^{-1/3} > 1/4$ since $\delta_v\ge 4$.
When $\delta_v = 2$, $v$ interacts only with $u$.
Therefore, the probability that $v$ interacts with others fewer than $\Delta$ times during $m/2$ interactions is greater than $1/4$.
When $\delta_u > 4$, let $X_u$ be a binomial random variable that represents the number of times that $u$ interacts with others excluding $v$ and $w$ within $m/2$ interactions.
In a similar way to $X_v$, the probability that $u$ interacts with others fewer than $\Delta-1$ times during $m/2$ interactions is greater than $1/4$.
Let $Y \sim \text{Bi}(m/2, 4/m)$ be a binomial random variable that represents the number of times that the edges $\{(u,v), (v,u), (v,w), (w,v)\}$ are chosen within $m/2$ interactions.
Using Lemma~\ref{Arisu_upper}, $\Pr(Y \geq 2) = 1 - \Pr(Y \leq 1) = 1 - \Pr(Y \leq (1-1/2)E[Y]) \geq 1 - e^{-1/4} \geq 1/5$.
It is sufficient that ii) happens first and iii) happens last for iii) to occur after ii).
The probabilities of ii) and iii) occurring are $1/2$ each, thus the probability that iii) occurs after ii) is no less than $1/4$.
The probability that the stamp in ii) differs from the stamp in iii) is $1/2$.
Therefore, the probability that the collision is detected by iii) after i) is $\Pr(X_v \leq \Delta-1) \Pr(X_u \leq \Delta-2) \Pr(Y \geq 2) \cdot 4^{-1} \cdot 2^{-1} > 640^{-1}$.
If any agents' colors changed during the $m/2$ interactions, it is likely that some agent detected collisions during these interactions.
Thus, $E_k \leq m/(2k) + m/2 + (1 - 640^{-1})E_k$.
Therefore, $E_k \leq 320m(1 + 1/k)$.
\end{proof}

\begin{lemma}\label{twohoprandom:collision}
When an agent $u$ interacts with another agent $v$ and a color collision is detected, the probability that neither of them makes new stamp collisions nor color collisions is no less than $1 - 1/(2N^2\Delta)$. 
\end{lemma}

\begin{proof}
When $u$ and $v$ regenerate colors, the new colors are independently and uniformly random.
Thus, the probability that $u$'s new color causes a collision with any colors in other agents' $\prev$ arrays is no more than $(n-1)\Delta/(\khopin)$.
Also, the probability that $u$'s new color causes a collision with other agents' colors is $(n-1)/(\khopin)$.
Since $v$ behaves similarly to $u$, the probability that neither $u$ nor $v$ causes collisions is no less than $1 - 2(n-1)\Delta/(\khopin) - 2(n-1)/(\khopin) \ge 1 - (2N^2\Delta)^{-1}$.
\end{proof}

\begin{lemma}\label{twohoprandom:ncc}
Let $C_0 \in \Call(\PLRU)$, and let $\Xi_\PLRU(C_0) = C_0, C_1, \dots$.
The number of interactions until the configuration reaches $\Ncc$ is $O(mn)$ with probability $1 - o(1)$.
\end{lemma}

\begin{proof}
In any configuration, the number of agents having color collisions is no more than $n$.
Using Lemma~\ref{twohoprandom:ek}, and applying Markov's inequality, the probability that any agent does not detect the collision within $1280m$ interactions is at least $1/2$, since $E_k \leq 320m(1 + 1/k) \leq 640m$ and $\Pr(Y \geq 2 \cdot 640m) \leq 1/2$, where $Y$ is a random number that the number of interactions until a color collision is detected.
Let $x (> 0)$ be the number of detected collisions until the configuration reaches $\Ncc$.
We denote the event $A$ that detects collisions during $1280m$ interactions.
The probability of success for $A$ is more than $1/2$.
Let $X_i \sim \text{Geom}(1/2)$ be a random number that the number of trials until $A$ succeeds, using the geometric distribution.
Consider the number of times until $A$ succeeds $x$ times as the sum of independent geometric distributions $X = \sum_{i=2}^{x}X_i$.
From Lemma~\ref{Janson_lower},
$\Pr(X \geq 4(n/x)E[X]) \leq e^{-x(4n/x-1-\log_e{(4n/x)})}$.
Here, we denote $f(x) = -x(4n/x-1-\log_e{(4n/x)})$, $(1 \leq x \leq n)$.
$f'(x) = \log_e{(4n/x)} > \log_e{4n/n} > 0$.
Thus, $f(x)$ is monotonically increasing.
$f(n) = n(-3+\log_e{4}) < -n$, thus, $e^{f(x)} < e^{-n}$.
Therefore, in less than $4(n/x)E[X] = 8n$ trials of $A$, $A$ succeeds $x$ times, and its probability is more than $1 - e^{-n}$.
Thus, the probability that the collision is detected more than $x$ times during $8n \cdot 1280m = 10240mn$ interactions is more than $1 - e^{-n}$.
Since the probability that colors are regenerated and do not re-collide is greater than $1 - (2N^2\Delta)^{-1}$ from Lemma~\ref{twohoprandom:collision}, the probability that colors do not collide after $x$ times of regeneration is greater than $1 - x/(2N^2\Delta) \geq 1 - (2N\Delta)^{-1}$.
Therefore, the probability that the configuration reaches $\Ncc$ during $10240mn = O(mn)$ interactions is at least $1 - e^{-n} - (2N\Delta)^{-1} = 1 - o(1)$.
\end{proof}

\begin{lemma}\label{twohoprandom:scol}
Let $C_0 \in \Ncc$, and let $\Xi_\PLRU(C_0) = C_0, C_1, \dots$.
The number of interactions until the configuration reaches $\Scol$ is $O(m\log{n})$ with probability $1 - o(1)$.
\end{lemma}

\begin{proof}
In any configuration within $\Ncc$, if some agents $u$ and $v$ have a stamp collision, $u$ and $v$ regenerate colors when they interact.
At this time, according to Lemma~\ref{twohoprandom:collision}, the probability that the number of stamp collisions decreases is no less than $1 - (2N^2\Delta)^{-1}$.
If all agents have stamp collisions, all agents must regenerate colors at least once.
When all agents change colors at least once, there are $2(n-1)$ instances or fewer of color regeneration.
The probability that the new colors do not cause new collisions is no less than $1 - 2(n-1)/(2N^2\Delta) \geq 1 - (N\Delta)^{-1}$.
It is sufficient for all pairs to interact at least once to regenerate all agents' colors.
The probability that any edge does not interact during $2m\log{n}$ interactions is less than $(1 - 2/m)^{2m\log{n}} < e^{-4\log{n}} < 2^{-4\log{n}} < n^{-4}$.
Thus, the probability that all pairs interact is more than $1 - m/n^4 > 1 - 1/n^2$ by the union bound.
Therefore, the probability that any configuration in $\Ncc$ reaches $\Scol$ during $2m\log{n}$ interactions is $1 - (N\Delta)^{-1} - n^{-2} = 1 - o(1)$.
\end{proof}

\begin{lemma}\label{twohoprandom:conv}
Let $C_0 \in \Call(\PLRU)$, and let $\Xi_\PLRU(C_0) = C_0, C_1, \dots$.
The number of interactions until the configuration reaches $\Scol$ is $O(mn)$ steps both in expectation and with high probability.
\end{lemma}

\begin{proof}
From Lemma~\ref{twohoprandom:ncc} and Lemma~\ref{twohoprandom:scol},
the time complexity for any configuration in \(\Call(\PLRU)\) to reach $\Scol$ is $O(mn)$, and this probability is at least $1 - o(1)$.
Thus, $\Pr(\exists i \in O(mn) : C_i \in \Scol) = 1 - o(1) = \Omega(1)$ holds.
Let $A = \max_{C \in \Scol} \ECT_\PLRU(C, \Scol)$.
We have $A = O(mn) + (1 - \Omega(1))A$, and solving this gives $A = O(mn)$.
\end{proof}

From Lemma~\ref{twohoprandom:closure} and Lemma~\ref{twohoprandom:conv}, we have the following theorem.

\twohoprandomteiri*

\subsection{\Deterministic~Protocol}\label{subsec:DLRU}
In this subsection, we propose a \deterministic~self-stabilizing two-hop coloring protocol, $\PDLRU$, mirroring the strategy used in the randomized protocol $\PLRU$. The primary distinction lies in the methods of generating colors and stamps. Unlike in $\PLRU$, where random numbers are utilized, the use of random numbers is prohibited in this context. Instead, we simulate the generation of random numbers through interactions.

Consider that each agent has a variable $\rand$ initialized at $1$. During an interaction, an agent $u$ acts as either the initiator or the responder, each with equal probability. If $u$ is the initiator, it updates $u.\rand$ to $2u.\rand$; if a responder, to $2u.\rand + 1$. Repeating this process $x$ times yields a random number within the range $[2^x, 2^{x+1}-1]$. An agent can determine if it has completed $x$ interactions since its $\rand$ remains less than $2^x$ until it surpasses this threshold. The randomness of these numbers is ensured by the uniform random scheduler, which impartially designates $u$ as either initiator or responder.

However, a challenge arises when all agents concurrently generate random numbers using this method; the random numbers generated may not be independent, as some interactions may overlap across multiple agents' generations. To ensure independence, we stipulate that only one agent should utilize an interaction at each step. This is achieved by:
i) employing a normal coloring protocol to establish a hierarchical relationship among adjacent agents,
ii) allowing agents to generate random numbers via interactions only if they are of higher rank compared to the agents they interact with, and recording these random numbers,
iii) enabling high-rank agents to provide random numbers to their lower-rank counterparts when needed.

We first introduce a \deterministic~self-stabilizing normal coloring protocol, followed by the \deterministic~self-stabilizing two-hop coloring protocol.

\subsubsection{Normal Coloring}
We propose a \deterministic~self-stabilizing normal coloring silent protocol $\PNC$.
This protocol colors adjacent agents to different colors.
Given an upper bound $N$, $\PNC$ converges to safe configurations within $O(mn\log{n})$ steps both in expectation and with high probability, and uses $O(\log{N})$ bits of memory on arbitrary graphs.
The mechanism of $\PNC$ is the same as the self-stabilizing leader election protocol on complete graphs proposed by Cai~\etal\cite{Cai2012}.

\begin{algorithm}[!htb]
    \caption{Self-Stabilizing Normal Coloring Protocol $\mathcal{P}_{NC}$}
    \label{Protocol:ssnc}
    \SetKwInOut{OutFunc}{Output function $O$}
    \setcounter{AlgoLine}{0}
    \OutFunc{each agent $a$ outputs $a.\clr$}
    \nonl\when{an initiator $a_0$ interacts with a responder $a_1$}{
        \If{$a_0.\clr = a_1.\clr$}{
            $a_1.\clr \gets a_1.\clr + 1 \bmod N$\;
        }
    }
\end{algorithm}

$\PNC$ is defined by algorithm~\ref{Protocol:ssnc}.
In $\PNC$, an agent $a$ has only one variable $a.\clr\in\{0,\dots,N-1\}$ representing the agent's color. 
$\PNC$ is very simple: when an initiator $a_0$ interacts with a responder $a_1$, if $a_0$ and $a_1$ are the same color, $a_1$ updates $a_1.\clr$ to $(a_1.\clr+1) \bmod N$. 
Each agent outputs $a.\clr$.

We define the safe configurations $\NC=\{C\in \Call(\PNC)\mid \forall u, \forall v \in V: (u,v) \in E \Rightarrow C(u).\clr \ne C(v).\clr\}$.

A color $c$ is unused in a configuration $C_i\ (i\ge 0)$ for any execution $\Xi(C_0)=C_0,\dots,C_i,\dots$, if and only if $\forall j \in [0,i], \forall v \in V: C_j(v).\clr \ne c$ holds.
A color $c$ is used in a configuration $C$ if and only if there exists an agent $v \in V$ satisfying $C(v).\clr = c$.
Note that once a color has become used, it will not revert to unused since an initiator's color does not change at each interaction.

\begin{lemma}[Closure]\label{nc:closure}
Let $C_0 \in \NC$, and let $\Xi_{\PNC}(C_0) = C_0, C_1, \dots$.
There is no configuration $C_i \notin \NC\ (i \ge 0)$.
\end{lemma}

\begin{proof}
Any configuration $C \in \NC$ satisfies $\forall u, \forall v \in V: (u,v) \in E \Rightarrow C(u).\clr \ne C(v).\clr$.
Thus, there is no agent changing its color.
Therefore, this lemma holds.
\end{proof}

\begin{lemma}\label{nc:mlogn}
The probability that all edges interact during $3m\log{n}$ interactions is at least $1 - n^{-4}$.
\end{lemma}

\begin{proof}
Let $X$ be a random number that the number of interactions that an edge interacts during $3m\log{n}$ interactions by $X\sim Bi(3m\log{n}, 2/m)$.
The probability that $X$ is greater than $0$ is 
$\Pr(X \ge 1) = 1 - \Pr(X = 0) \ge 1 - (1 - 2/m)^{3m\log{n}} \ge 1 - e^{-6\log{n}} \ge 1 - n^{-6}$.
Thus, the probability that all edges interact during $3m\log{n}$ interactions is at least $1 - m/n^6 \ge 1 - n^{-4}$ by the union bound.
\end{proof}

\begin{lemma}[Convergence]\label{nc:con}
Let $C_0 \in \Call(\PNC)$, and let $\Xi_\PNC(C_0) = C_0, C_1, \dots$.
The number of interactions until the configuration reaches $\NC$ is $O(mn\log{n})$ steps both in expectation and with high probability.
\end{lemma}

\begin{proof}
Let $Col$ denote a set of colors used in $\Xi_\PNC(C_0)$.
We introduce a directed graph $G_{assign} = (Col, E_{assign})$ to represent color assignment.
An edge $(c,(c+1 \mod N)) \in E_{assign}$ iff a color $c+1$ is assigned at line 2 in Algorithm~\ref{Protocol:ssnc} in $\Xi_\PNC(C_0)$.

First, we show $G_{assign}$ is acyclic.
$G_{assign}$ can be cyclic only if $n=N$, $Col = \{0,1,\ldots,N-1\}$ and $E_{assign} = \{(c,(c+1 \mod N))|c \in \{0,1,\ldots,N-1\})$. If all the colors in $Col$ are used in $C_{0}$, $C_{0} \in \NC$ and no color is assigned in the execution. That is $E_{assign} = \emptyset$ and $G_{assign}$ is acyclic. Otherwise, there are colors in $Col$ that are not used in $C_{0}$. Let $c'$ be the color that is used last in the execution. When $c'$ is assigned, all the $n$ colors are used and this implies $\NC$ is reached. Therefore, no color will be assigned after that, and there is no outgoing edge from $c'$ in $G_{assign}$. That is, $G_{assign}$ is acyclic. 

Since $G_{assign}$ is acyclic, it is divided into weakly connected components, where each component forms a line segment $c_{i},c_{i}+1,\ldots,d_{i}$. This means no interaction assigns a color $c_{i}$ or $(d_{i} \mod N)$.
From Lemma~\ref{nc:mlogn}, the probability that all edges interact during $3m\log{n}$ interactions is at least $1 - n^{-4}$. In the first $3m\log{n}$ interactions, color collisions with $c_{i}$ disappear for each component with probability at least $1 - n^{-4}$, and then, in the next $3m\log{n}$ interactions, color collisions with $c_{i}+1$ disappear similarly. Repeating this at most $n$ times, all the color collisions disappear and reaches $\NC$ with probability at least $1 - n^{-3}$.

The number of interactions to reach $NC$ with probability at least $1 - n^{-3}$ is $O(nm\log{n})$, and the expected number $E_{NC}$ of interactions to reach $NC$ is $E_{NC} = O(nm\log{n}) + n^{-3}\cdot E_{NC}$, Solving this, we have $E_{NC} = O(nm\log{n})$.
\end{proof}

In $\NC$, agents do not change their variables.
Thus, $\PNC$ is a silent protocol.

The following theorem holds from Lemma~\ref{nc:closure} and Lemma~\ref{nc:con}:
\begin{theorem}
Given the upper bound $N$, $\PNC$ is a \deterministic~self-stabilizing normal coloring silent protocol, and converges to safe configurations within $O(mn\log{n})$ steps both in expectation and with high probability.
\end{theorem}

\subsubsection{Two-Hop Coloring}  
We propose a \deterministic~self-stabilizing two-hop coloring protocol $\PDLRU$.  
Given the upper bound $N$ and $\Delta$, $\PDLRU$ converges to safe configurations within $O(mn+m\Delta\log{N})$ both in expectation and with high probability and uses $O(\Delta\log{N})$ bits of memory.

The main point of difference between $\PLRU$ and $\PDLRU$ is the way to generate random numbers.

In $\PDLRU$, an agent $a$ has $9$ variables:  
$a.\nc\in \{0,\dots ,N-1\}$,
$a.\rand\in \{0,2^{\khop}-1\}^{\randsize}$,
$a.\gen \in \{1,2^{\khop+1}-1\}^{\randsize}$,
$a.\cur\in \{1,\dots,\randsize+1\}$,
$a.\gencur\\\in \{1,\dots,\randsize\}$,
$a.\hop \in \{0, \dots, 2^\khop - 1\}$,
$a.\prev\in \{0, \dots, 2^\khop - 1\}^\Delta$,
$a.\stamp \in \{0,1\}^\Delta$, and  
$a.\idx \in \{0,\dots,\Delta\}$.  
The variable $a.\nc$ is a read-only variable representing a normal coloring color.  
The variable $a.\rand$ is an array that stores $\randsize$ generated random numbers, and $a.\cur$ represents the index of $a.\rand$ where $a$ should get a random number from $a.\rand$ when necessary.  
The variables $a.\gen$ and $a.\gencur$ are used in generating random numbers as described later.  
The variables $a.\hop$, $a.\prev$, $a.\stamp$, and $a.\idx$ have the same roles as in $\PLRU$.  
\begin{algorithm}[!htb]
    \caption{Self-Stabilizing Two-Hop Coloring Protocol $\PDLRU$}
    \label{Protocol:ss2hop}
    \SetKwInOut{OutFunc}{Output function $O$}
    \SetKwProg{Fn}{function}{:}{}
    \SetKwFunction{GenClr}{Generate\_Color}
    \SetKwFunction{GenBit}{Generate\_Bit}
    \SetKwFunction{GetRandomBit}{Get\_Bit}
    \SetKw{Downto}{downto}
    \setcounter{AlgoLine}{0}
    \OutFunc{Each agent $a$ outputs $a.\hop$.}
    \nonl\when{an initiator $a_0$ interacts with a responder $a_1$}{
        $\PNC$\;
        \If{$\exists i\in\{0,1\}:a_i.\nc > a_{1-i}.\nc$}{
            \nonl// Reset due to inconsistency\;
            \lIf{$a_i.\gen[a_i.\gencur]\ge 2^\khop$}{
                $(a_i.\gencur, a_i.\gen[1])\gets (1,1)$
            }
            $a_i.\gen[a_i.\gencur] \gets 2a_i.\gen[a_i.\gencur]+i$\;
            \If{$a_i.\gen[a_i.\gencur]\ge 2^\khop$}{
                \uIf{$a_i.\gencur = \randsize$}{
                    \For{$j\gets 1\ \KwTo\ \randsize$}{
                        $a_i.\rand[j]\gets a_i.\gen[j]-2^\khop$\;
                    }
                    $(a_i.\cur,a_i.\gencur)\gets (1,1)$\;
                }\lElse{
                    $a_i.\gencur\gets a_i.\gencur+1$
                }
                $a_i.\gen[a_i.\gencur]\gets 1$\;
            }
        }
        \Forall{$i \in \{0,1\}$}{
            $a_i.\idx \gets 0$\;
            \For{$j \gets 1\ \KwTo\ \Delta$}{
                \If{$a_i.\prev[j] = a_{1-i}.\hop$}{
                    $a_i.\idx\gets j$\;
                    \textbf{break}\;
                }
            }
        }
        \GenClr()\;
        \Forall{$i \in \{0,1\}$}{
            \lIf{$a_i.\idx =0$}{
                $a_i.\idx \gets \Delta$
            }
            \For{$j \gets a_i.\idx-1\ \Downto\ 1 $}{
                $(a_i.\prev[j+1],a_i.\stamp[j+1])\gets (a_i.\prev[j],a_i.\stamp[j])$\;
            }
        }
        \GenBit()\;
    }
    \nonl\Fn{\GenClr{}}{
        \If{$a_0.\idx>0 \wedge a_1.\idx>0 \wedge a_0.\stamp[a_0.\idx] \ne a_1.\stamp[a_1.\idx])$}{
                \If{$\exists i\in \{0,1\}:a_i.\nc>a_{1-i}.\nc \wedge a_i.\cur + 1 \le \randsize$}{
                    $(a_0.\hop,a_1.\hop)\gets (a_i.\rand[a_i.\cur],a_i.\rand[a_i.\cur+1])$\;
                    $(a_0.\idx,a_1.\idx,a_i.\cur) \gets (0,0,a_i.\cur+2)$\;
                }
            }
    }
    \nonl\Fn{\GetRandomBit{}}{
        \lIf{$a_0.\nc < a_1.\nc$}{
            \textbf{return} $0$
        }\lElse{
            \textbf{return} $1$
        }
    }
    \nonl\Fn{\GenBit{}}{
        $(a_0.\prev[1],a_1.\prev[1],a_0.\stamp[1],a_1.\stamp[1]) \gets (a_1.\hop,a_0.\hop,\GetRandomBit(),\GetRandomBit())$\;
    }
\end{algorithm}

$\PDLRU$ is given by Algorithm~\ref{Protocol:ss2hop}, and consists of 6 parts:  
i) normal coloring, ii) generating random numbers, iii) reading memory, iv) collision detection, v) saving colors, and vi) stamping.

i) Normal coloring (line 1) executes $\PNC$ to do normal coloring.

ii) Generating random numbers (lines 2--11) aims to create random numbers.  
This part generates random numbers. For each interaction, at most one agent generates $1$ bit of a random number with $2^\khop$ bits, and the generated random numbers (with deleting the most significant bit) are copied to an array rand every time each agent generates $\Delta$ random numbers. A variable $\gen$ stores generated random numbers and $\gencur$ points to a currently generating number.
After the configuration reaches $\NC$, all adjacent agents have different colors.  
In other words, there exists $i\in\{0,1\}$ satisfying $a_i.\nc > a._{1-i}.\nc$ when $a_0$ interacts with $a_1$.  
When this inequality holds, $a_i$ appends one bit ($0$ if initiator, $1$ if responder) to currently generating random number (line 4). If the generating number reaches the limit (line 5), $\gencur$ is incremented or the generated numbers are copied to $\rand$ and $\cur$ and $\gencur$ are reset when $\Delta$ random numbers have been generated (line 6 -- 11). When copying, the most significant bit is deleted so that a random number stays in a range $[0,2^{\khop}-1]$ (line 8).
This protocol is self-stabilizing, so there may exist non-generated numbers in $u.\rand$ in some configurations.  
However, when $u$ updates $u.\rand$ two or more times, all elements of $u.\rand$ are generated random numbers.

iii) Reading memory (lines 12--17) is the same as in $\PLRU$.

iv) Collision detection (lines 18, and 24--28) aims to regenerate colors when both $a_0$ and $a_1$ remember each other, and if their stamps are different.  
The agents use the generated random numbers described above.  
If there is $i\in\{0,1\}$ satisfying $a_i.\nc > a_{1-i}.\nc$ and $a_i.\cur+1\le \randsize$, $a_0.\hop$ and $a_1.\hop$ are updated to new colors $a_i.\rand[a_i.\cur],a_i.\rand[a_i.\cur+1]$, and $a_i$ updates $a_i.\cur$ to $a_i.\cur+2$ (lines 25--27).  
Note that if $a_i.\cur+1> \randsize$ or $a_0.\nc=a_1.\nc$ holds, agents do not regenerate colors.

v) Saving colors (lines 19--22) is the same as in $\PLRU$.

vi) Stamping (lines 23, and 28--30) aims to remember stamps to detect collisions.  
This is the same as in $\PLRU$ except for the way of generating a stamp.  
Generating a stamp (lines 29--30) uses the $\nc$ and the roles of the initiator and the responder.  
When $a_0$ interacts with $a_1$, if $a_0.\nc < a_1.\nc$ holds, the stamp becomes $0$; otherwise, the stamp becomes $1$.  
After the configuration reaches $\NC$, the probability that the stamp becomes $0$ or $1$ is $1/2$ respectively since the roles of the initiator and the responder are decided by the uniform random scheduler.

\begin{lemma}\label{twohop:uniform}  
All random numbers that agents generate are uniform and independent.  
\end{lemma}

\begin{proof}  
Since interactions are decided by a uniform random scheduler, the probability that an interacting agent is an initiator (resp. responder) also follows a uniform random scheduler at each interaction.  
Thus, the random numbers generated by the roles of the initiator and the responder are uniform random numbers.  
When agents generate random numbers, either the initiator or the responder uses the interaction.  
Thus, no two generated random numbers are related.  
Therefore, generated random numbers are uniform and independent.  
\end{proof}

We use $\Ncc$ and $\Scol$ described in $\PLRU$.

\begin{lemma}[Closure]\label{twohop:closure}  
Let $C_0\in \Scol$, and $\Xi_\PDLRU(C_0)=C_0,C_1,\dots$.  
$\forall i\in \mathbb{N} :C_i \in \Scol$ holds and each agent's color does not change from $C_0$.  
\end{lemma}

\begin{proof}  
This lemma follows directly from Lemma~\ref{twohoprandom:closure}.  
\end{proof}

We call an agent $u$ the rand-gen agent whose $\nc$ is not the minimum among all adjacent agents' $\nc$.

\begin{lemma}\label{twohop:rand1}  
Let $C_0\in \NC$, and $\Xi_\PDLRU(C_0)=C_0,C_1,\dots$. 
The number of interactions until all rand-gen agents' $\rand$ are updated at least once is $O(m\Delta\log{N})$ steps both in expectation and with high probability.  
\end{lemma}

\begin{proof}  
First, we show the number of interactions until all rand-gen agents generate one or more random numbers.  
A rand-gen agent $u$ needs to participate in $\khop$ interactions satisfying $\exists v\in V: u.\nc > v.\nc$ to create a random number.  
Let $X\sim Bi(2m \khop, 2/m)$ be a random number that the number of interactions where $u$ interacts during $2m\khop$ interactions.  
From Lemma~\ref{Arisu_upper}, the probability that $u$ interacts no less than $\khop$ during $2m\khop$ interactions is $\Pr(X\ge \khop)=1-\Pr(X\le \khop)\ge 1-e^{-2\cdot 2\khop/4}\ge 1-e^{-\khop}\ge 1-e^{-\log{\khopin}}>1-(\khopin)^{-1}$.  
Thus, by the union bound, the probability that each rand-gen agent interacts no less than $\khop$ times during $2m\khop$ is no less than $1-(8N^2\Delta^2)^{-1}$.  
Therefore, the probability that each rand-gen agent generates $\randsize$ random numbers during $2m\Delta\khop=O(m\Delta\log{N})$ interactions is no less than $1-(8N^2\Delta)^{-1}$ by the union bound.  
\end{proof}

\begin{lemma}\label{twohop:rand2}  
Let $C_0\in \NC$, and $\Xi_\PDLRU(C_0)=C_0,C_1,\dots$.  
The number of interactions until all rand-gen agents' $\rand$ are generated random numbers is $O(m\Delta\log{N})$ steps both in expectation and with high probability.  
\end{lemma}

\begin{proof}  
For any rand-gen agent $u$, when $u$ updates $u.\rand$ for the first time in this execution, $u$ starts generation from the first element of $u.\gen$.  
When $u$ updates $u.\rand$ for the second time in this execution, all elements of $u.\rand$ are generated in this execution.  
Therefore, from Lemma~\ref{twohop:rand1}, the number of interactions that each rand-gen agent updates its $\rand$ at least twice is $2\cdot O(m\Delta\log{N})=O(m\Delta\log{N})$ and the probability is no less than $1-(4N^2\Delta)^{-1}$.  
\end{proof}

\begin{lemma}
Let $C_0\in \Call(\PDLRU)$, and $\Xi_{\PDLRU}(C_0)=C_0,C_1,\dots$.
The number of interactions until the configuration reaches $\Scol$ is $O(m(n\log{n}+\Delta\log{N}))$ steps both in expectation and with high probability.
\end{lemma}

\begin{proof}
Let \( C_r \) be the configuration where $\rand$ of all rand-gen agents first becomes one that was generated within this execution.
From Lemma~\ref{nc:con} and Lemma~\ref{twohop:rand2}, the configuration becomes $C_r$ within $O(m(n\log{n}+\Delta\log{N}))$ steps with high probability.
Consider the configuration where the $\rand$ of all rand-gen agents has been regenerated at least once.  
Note that this event occurs within $O(m\Delta\log{N})$ steps with high probability, as stated in Lemma~\ref{twohop:rand1}.  
At this point, since the $\cur$ of all agents has been reset at least once after $C_r$, the condition $u.\cur + 1 < \Delta$ holds with high probability until $\Xi$ has converged to $\Scol$.  
Therefore, by Lemma~\ref{twohoprandom:conv}, the configuration reaches $\Scol$ within $O(m(n\log{n}+\Delta\log{N}))$ steps with high probability.
Let $A=\max_{C\in \Call(\PDLRU)}\ECT(C,\Scol)$.
$\Pr(\exists i\in O(m(n\log{n}+\Delta\log{N})))=1-o(1)=\Omega(1)$ holds, thus, $A=O(m(n\log{n}+\Delta\log{N}))+(1-\Omega(1))A$ holds.  
Solving this gives $A=O(m(n\log{n}+\Delta\log{N}))$.  
\end{proof}

To analyze time complexity, we assume the following assumption.

\begin{assumption}\label{twohop:ass}  
All elements of all rand-gen agents' $\rand$ are generated in this execution, and each rand-gen agent $u$ satisfies $u.\cur+1\le \randsize$ when $u$ regenerates colors in this execution. 
\end{assumption}

\begin{lemma}\label{twohop:ek}  
Once Assumption~\ref{twohop:ass} is satisfied during an execution, $E_k \le 320m(1/k+1)$ will hold.
\end{lemma}

\begin{proof}  
Once Assumption~\ref{twohop:ass} is satisfied during an execution, this lemma will be directly derived from Lemma~\ref{twohoprandom:ek}.
\end{proof}

\begin{lemma}\label{twohop:collision}  
Once Assumption~\ref{twohop:ass} is satisfied during an execution, if an agent $u$ interacts with another agent $v$ and a color collision is detected, the probability that neither agent causes new stamp collisions and color collisions will be no less than $1-1/(2N^2\Delta)$.  
\end{lemma}

\begin{proof}  
Once Assumption~\ref{twohop:ass} is satisfied during an execution, from Lemma~\ref{twohop:uniform}, the generated colors are uniform and independent.  
Thus, this lemma will be directly derived from Lemma~\ref{twohoprandom:collision}.
\end{proof}

\begin{lemma}\label{twohop:ncc}  
Let $C_0\in \NC$, and $\Xi_\PDLRU(C_0)=C_0,C_1,\dots$. 
Once Assumption~\ref{twohop:ass} is satisfied during an execution, the number of interactions until the configuration reaches $\Ncc$ will be $O(mn)$ steps with probability $1-o(1)$.  
\end{lemma}

\begin{proof}  
Once Assumption~\ref{twohop:ass} is satisfied during an execution, this lemma will be directly derived  from Lemma~\ref{twohoprandom:ncc} and Lemma~\ref{twohop:collision}.  
\end{proof}

\begin{lemma}\label{twohop:scol}  
Let $C_0\in \NC \cap \Ncc$, and $\Xi_\PDLRU(C_0)=C_0,C_1,\dots$.  
Once Assumption~\ref{twohop:ass} is satisfied during an execution, the number of interactions until the configuration reaches $\Scol$ will be $O(m\log{n})$ steps with probability $1-o(1)$.  
\end{lemma}

\begin{proof}  
Once Assumption~\ref{twohop:ass} is satisfied during an execution, this lemma will be directly derived from Lemma~\ref{twohoprandom:scol} and Lemma~\ref{twohop:collision}.  
\end{proof}

\begin{lemma}\label{twohop:conv}  
Let $C_0\in \NC$, and $\Xi_\PDLRU(C_0)=C_0,C_1,\dots$.
Once Assumption~\ref{twohop:ass} is satisfied during an execution, the number of interactions until the configuration reaches $\Scol$ will be $O(mn)$ steps with probability $1-o(1)$.  
\end{lemma}

\begin{proof}  
Once Assumption~\ref{twohop:ass} is satisfied during an execution, this lemma will be directly derived from Lemma~\ref{twohop:ncc}, and Lemma~\ref{twohop:scol}.  
\end{proof}

\begin{lemma}\label{twohop:cnt}  
Let $C_0\in \NC$, and $\Xi_\PDLRU(C_0)=C_0,C_1,\dots$.
Once Assumption~\ref{twohop:ass} is satisfied during an execution, the number of times that each agent regenerates colors will be no more than $\lfloor \Delta/2\rfloor$ with probability $1-o(1)$.  
\end{lemma}

\begin{proof}  
We consider the situation after Assumption~\ref{twohop:ass} has been satisfied during an execution.
An agent $u$ has $\delta_u/2$ neighbors.  
When $u$ interacts with an agent $v$ with the same color, $u$ and $v$ do not make new collisions with a probability of at least $1-(2N^2\Delta)^{-1}$ from Lemma~\ref{twohop:collision}.  
Thus, the number of times that $u$ regenerates colors is no more than $\delta_u/2$ with probability at least $1-\delta_u/(4N^2\Delta)\ge 1-N^{-2}$.  
Therefore, once Assumption~\ref{twohop:ass} is satisfied during an execution, the probability that the number of times each agent regenerates colors is no more than $\lfloor \Delta/2\rfloor$ will be no less than $1-N^{-1}$.  
\end{proof}

\begin{lemma}\label{twohop:allconv}  
Let $C_0\in \Call(\PDLRU)$, and $\Xi_\PDLRU(C_0)=C_0,C_1,\dots$.  
The number of interactions until the configuration reaches $\Scol$ is $O(m(n+\Delta\log{N}))$ steps both in expectation and with high probability.  
\end{lemma}

\begin{proof}  
From Lemma~\ref{nc:con}, the configuration reaches $\NC$ within $O(mn\log{n})$ steps both in expectation and with high probability.  
After that, we consider when the configuration still reaches $\NC$.  
From Lemma~\ref{twohop:cnt}, if all elements of all rand-gen agents' $\rand$ are generated in the execution, the number of times each agent regenerates colors is no more than $\lfloor \Delta/2 \rfloor$ with probability $1-o(1)$.  
When an agent $u$ interacts with $v$ and they generate colors, either $u$ or $v$ generates two colors.  
Thus, each rand-gen agent needs no more than $2\lfloor \Delta/2 \rfloor$ random numbers.  
After all rand-gen agents' $\rand$ are updated at least once, and after all elements of all rand-gen agents' $\rand$ become generated random numbers in the execution, each rand-gen agent regenerates no more than $\lfloor \Delta/2 \rfloor$ times with probability $1-o(1)$.  
Since $\rand$'s size is $\Delta$, rand-gen agents can regenerate colors at any time with high probability.  
From Lemma~\ref{twohop:rand2}, the number of interactions that all elements of all rand-gen agents' $\rand$ become generated random numbers in the execution is $O(m\Delta\log{N})$ steps both in expectation and with high probability.  
From Lemma~\ref{twohop:rand1}, the number of interactions that all rand-gen agents' $\rand$ are updated at least once again is $O(m\Delta\log{N})$ steps both in expectation and with high probability.  
After the above, any configuration belonging to $\NC$ reaches $\Scol$ within $O(mn)$ with probability $1-o(1)$ from Lemma~\ref{twohop:conv}.  
Therefore, the number of interactions until the configuration reaches $\Scol$ is $O(m(n+\Delta\log{N}))$ steps with probability $1-o(1)$.  
Let $A=\max_{C\in \Call(\PDLRU)}\ECT(C,\Scol)$.  
$\Pr(\exists i\in O(m(n+\Delta\log{N})))=1-o(1)=\Omega(1)$ holds, thus, $A=O(m(n+\Delta\log{N}))+(1-\Omega(1))A$ holds.  
Solving this gives $A=O(m(n+\Delta\log{N}))$.  
\end{proof}

The following theorem holds from Lemma~\ref{twohop:closure} and Lemma~\ref{twohop:allconv}.

\twohopteiri*

\section{Proofs and Analysis of Leader Election}

\LEindependent*

\begin{proof}  
When candidates generate $\iid$s, if both initiator and responder are candidates, the responder becomes a follower.  

When leaders generate $\type$s, if both initiator and responder are leaders, they move to the Global Reset phase.  

In both cases, one interaction cannot be used for two random numbers, which guarantees independence, and a uniform random scheduler determines the roles of an initiator and a responder with equal probability, which guarantees uniformity.  
\end{proof}

\subsection{Proofs of Expected Holding Time's Lemmas}

\LEholdcount*

\begin{proof}  
For any configuration in $\SDcol$, all adjacent agents of every agent have different colors.  
Thus, the probability that every agent encounters no less than $8\tau$ counting interactions while $\Gamma_0,\dots,\Gamma_{2m\tau-1}$ is less than $1-ne^{-\tau}$, from Lemma 3 of Sudo~\etal\cite{SameSpeedTimer}.  
This lemma holds since $8\tau\le \tbc/2$ holds.  
\end{proof}

\LEholdspreads*

\begin{proof}  
We show this lemma in the same way as Lemma 4 of Sudo~\etal\cite{SameSpeedTimer}.
Since the value of timer is propagated with decreasing uniformly through all the paths from $v$ to $u$ and propagation through the shortest path affect at the earliest, we analyze how the value if propagated through the shortest path.
Let $C_0(v).\timer_x=y$ and we denote a shortest path from $v$ to an agent $u$ by $(v_0,v_1,\dots,v_k)$ where $v_0=v$ and $v_k=u$ holds.  
For any $t\in[0,2m\tau]$, we define $v_{head}(t)$ as follows:  
$v_{head}(0)=v_0$.
When $v_{head}(t-1)=v_{l-1}$, $v_{head}(t)=v_l$ if $\Gamma_{t-1}\in \{(v_{l-1},v_l),(v_l,v_{l-1})\}$, otherwise $v_{head}(t)=v_{l-1}$.
We also define $J(t)$ as the total number of times $v_{head}(i)$ encounters counting interactions for $i\in[0,t]$.  
$v_{head}.\timer_x$ decreases when and only when $v_{head}$ moves or $v_{head}$ encounters counting interactions.  
Thus, we get $C_t(v_{head}).\timer_x\ge y-d-J(t)\ge \tbc/2-d-J(t)$.  
From Lemma 4 of Sudo~\etal\cite{SameSpeedTimer},  
$\Pr(v_{head}(2m\tau)=u)\ge 1-e^{-\tau}$, and $\Pr(J(2m\tau)<\tbc/2-d)\ge 1-e^{-\tau}$ holds.  
Since the number of candidates of $u$ is less than $n$, we get $\Pr(\forall v\in V:C_{2m\tau}(v).\timer_x>y-\tbc/2)\ge 1-2ne^{-\tau}$ by the union bound.  
\end{proof}

\LEholdvoneclean*

\begin{proof}  
\Lnotincreasing

We will show that the unique leader does not disappear from $C_0$ to $C_{2m\tau}$ with high probability.  
The unique leader disappears if and only if an old search virus lives and the search virus meets the current search virus or the leader.  
Here, we say the search virus is old if it existed when the leader became $\LL_0$.

Firstly, we consider the case where the leader begins generating its $\type$ for the first time during $\Gamma_0,\dots \Gamma_{2m\tau-1}$ (\ie, the leader's $\timer_\mathrm{E}$ became $0$).  
At this time, the leader becomes $\LL_1$.  
If the leader finishes generating its $\type$ in $\Gamma_i\ (i\in [0,2m\tau))$, the probability that the leader's $\timer_\mathrm{V}$ does not decrease no less than $\tbc/2$ from $C_{i+1}$ to $C_{2m\tau}$ is at least $1-ne^{-\tau}$, from Lemma~\ref{leaderelection:holding:count}.  
Thus, the leader's $\timer_\mathrm{E}$ is always $2\tbc$ and does not start generating $\type$ again with probability at least $1-ne^{-\tau}$.  
Therefore, $C_{2m\tau}$ is in $\Vonly \cup \Vmake$ with probability at least $1-ne^{-\tau}$. 
No kill virus is created since $C_0\in \Vclean$ and the leader does not start generating $\type$ no less than twice.  
Thus, the probability that $C_{2m\tau}\in \SLE \cap (\Lvone \cap (\Vonly \cup \Vmake) \cap \Ehalf)$ holds if the leader starts generating $\type$ from $C_0$ to $C_{2m\tau}$ is at least $1-3ne^{-\tau}$ by the union bound.

Next, we consider the case where the leader does not start to generate its $\type$ during $\Gamma_0,\dots,\Gamma_{2m\tau-1}$.  
Since $C_0\in \Vclean$, there is no search virus in $C_{2m\tau}$ and no kill virus.  
Thus, the probability that $C_{2m\tau} \in \SLE \cap \Vclean$ is at least $1-2ne^{-\tau}$ by the union bound.

From the above, this lemma is proved.  
\end{proof}

\LEholdvonemakeehalf*

\begin{proof}  
\Lnotincreasing

There is no search virus in $C_0$.  
If the leader does not finish generating its $\type$ during $\Gamma_0,\dots \Gamma_{2m\tau-1}$, $\Pr(C_{2m\tau}\in \SLE \cap \Lvone \cap \Vmake \cap \Ehalf)\ge 1-2ne^{-\tau}$ holds by the union bound.  
Note that there is no kill virus in the configurations.

If the leader finishes generating its $\type$ in $\Gamma_i\ (i\in[0,2m\tau))$, the probability that the leader's $\timer_\mathrm{V}$ does not decrease no less than $\tbc/2$ is at least $1-ne^{-\tau}$ from Lemma~\ref{leaderelection:holding:count}.  
Thus, the leader's $\timer_\mathrm{E}$ is always $2\tbc$ and does not start generating $\type$ again.  
Therefore, the probability that $C_{2m\tau}\in \SLE \cap \Lvone \cap \Vonly \cap \Ehalf$ is at least $1-3ne^{-\tau}$ by the union bound.  
Note that there is also no kill virus in the configurations.

From the above, this lemma is proved.  
\end{proof}

\LEholdvoneonlyehalf*

\begin{proof}  
\Lnotincreasing

If the search virus does not disappear during $\Gamma_0,\dots,\Gamma_{2m\tau-1}$, there exists an agent $v$ satisfying $C_i(v).\timer_\mathrm{E}=2\tbc$ for $i\in [0,2m\tau]$ since $\forall i\in[0,2m\tau],\exists v\in V:C_i(v).\timer_\mathrm{V}>0$ holds.  
The probability that $\forall i\in[0,2m\tau],\exists v\in V:C_i(v).\LF=\LL_1 \Rightarrow C_i(v).\timer_\mathrm{E}\ge \tbc/2$ and $C_{2m\tau}\in \Ehalf$ hold is at least $1-3ne^{-\tau}$ from Lemma~\ref{leaderelection:holding:Elambda}.

If all the search viruses disappear during $\Gamma_0,\dots,\Gamma_{2m\tau-1}$, the probability that the leader's $\timer_\mathrm{E}$ does not become $0$ is at least $1-ne^{-\tau}$ from Lemma~\ref{leaderelection:holding:count} since $C_0\in \Ehalf$.  
Thus, the leader does not start generating $\type$ and $C_{2m\tau}\in \Vclean$.

In both cases, since the new search virus is not created, no kill virus is created.  
Therefore, $C_{2m\tau}$ also belongs to $\KLzero$.

From the above, this lemma is proved by the union bound.  
\end{proof}

\subsection{Analysis of Expected Convergence Time}
We also define the sets of configurations to prove Lemma~\ref{leaderelection:conv:conv}:

$\KLhalf = \{C \in \SDcol \mid \exists v \in V:C(v).\timer_\KL\ge \tbc/2\}$

$\Bexists = \{C \in \SDcol \mid \exists v \in V:C(v).\LF = \Baby\}$

$\IDclear = \{C \in \SDcol \mid \forall v\in V:C(v).\iid = 1\}$

$\IDsame = \{C \in \SDcol \mid \forall v, \forall u\in V:C(v).\iid = C(u).\iid\}$

$\Fall = \{C \in \SDcol \mid \forall v \in V:C(v).\LF = \FF\}$

$\Ldupl = \{C \in \SDcol \mid |\{v \in V \mid C(v).\LF \in \{\LL_0,\LL_1\}\}|\ge 2\}$

$\Lexists = \Lone \cup \Ldupl$

$\Lvzero = \{C \in \SDcol \mid \exists v \in V:C(v).\LF = \LL_0\}$

$\Equa = \{C \in \SDcol \mid \forall v \in V:C(v).\LF \in \{\LL_0, \LL_1\} \Rightarrow C(v).\timer_\mathrm{E} \ge \tbc/2\}$

\begin{lemma}\label{leaderelection:conv:katei}  
Let $C_0\in \Scol$ and $\Xi_\PBC(C_0)=C_0,C_1,\dots$.  
If $\Pr(\exists i\in O(m\tau\log{n}):C_i\in \SLE)=\Omega(1)$ holds,  
$\max_{C\in \Scol}\ECT_\PBC(C,\SLE)=O(m\tau\log{n})$ holds.  
\end{lemma}

\begin{proof}  
Let $B=\max_{C\in \Scol}\ECT_\PBC(C,\SLE)$.  
If $\Pr(\exists i\in O(m\tau\log{n}):C_i\in \SLE)=\Omega(1)$ holds,  
$B=O(m\tau\log{n})+(1-\Omega(1))B$ holds.  
Solving this inequality gives $B=O(m\tau\log{n})$.  
\end{proof}

We will show the convergence by dividing into some steps shown in Table~\ref{table:conv}.

\begin{table}[!htb]  
    \caption{Convergence Steps in $\PBC$}  
    \label{table:conv}      
    \centering  
    \begin{tabular}{lccc}  
        \hline  
        \multicolumn{1}{c}{Convergence Step} & Convergence Time & Probability & Lemmas\\  
        \hline  
        $\Scol \rightarrow \SDcol$ & $m\log{n}$ & $1-o(1)$ & Lemma~\ref{leaderelection:conv:SDcol}\\  
        $\SDcol \rightarrow \KLzero \cup \KLhalf$ & $2340m\tau\log{n}$ & $1-o(1)$ & Lemma~\ref{leaderelection:conv:Scol}\\  
        $\KLzero \rightarrow \KLzero \cap \Vclean \cup \KLhalf$ & $2m\tau+4680m\tau\log{n}$ & $1-o(1)$ & Lemma~\ref{leaderelection:conv:KLzero}\\  
        $\KLzero \cap \Vclean \rightarrow \SLE \cup \KLhalf$& $4m\tau+18720m\tau\log{n}$ & $1-o(1)$ & Lemma~\ref{leaderelection:conv:KLzeroVclean}\\  
        $\KLhalf \rightarrow \KLzero \cap \Fall \cap \Vclean \cap \IDclear$ & $2340m\tau\log{n}$ & $1-o(1)$ & Lemma~\ref{leaderelection:conv:KLhalf}\\  
        $\KLzero \cap \Fall \cap \Vclean \cap \IDclear \rightarrow \SLE$ & $2m\tau+7020m\tau\log{n}$ & $1-o(1)$ & Lemma~\ref{leaderelection:conv:FalltoSLE}\\  
        \hline  
        $\Scol\rightarrow \SLE$ & $O(m\tau\log{n})$ & $1-o(1)$ & \\  
        \hline  
    \end{tabular}  
\end{table}

\begin{lemma}\label{leaderelection:conv:SDcol}  
Let $C_0\in \Scol$ and $\Xi_\PBC(C_0)=C_0,C_1,\dots$.  
The probability that the configuration reaches $\SDcol$ during $m\log{n}$ interactions is at least $1-n^{-3}$.  
\end{lemma}

\begin{proof}  
It is sufficient that all agents interact to reach $\SDcol$.  
The probability that an agent $u$ does not interact during $m\log{n}$ interactions is $(1-\delta_u/m)^{m\log{n}}<e^{-2\delta_u\log{n}}<e^{-\log{n^4}}<2^{-\log{n^4}}=n^{-4}$ since $\delta_u\ge 2$.  
Thus, the probability that all agents interact during $m\log{n}$ interactions is at least $1-n^{-3}$ by the union bound.  
Note that once the execution reaches $\SDcol$, the configurations will not deviate from $\SDcol$ after that.  
\end{proof}

\begin{lemma}\label{leaderelection:conv:KLhalf}
Let $C_0\in \KLhalf$ and $\Xi_\PBC(C_0)=C_0,C_1,\dots$.  
The number of interactions until the configuration reaches $\KLzero \cap \Fall \cap \Vclean \cap \IDclear$ is less than $2340m\tau\log{n}$ with probability at least $1-3ne^{-\tau}-e^{-\tau}-n^{-1}$.  
\end{lemma}

\begin{proof}  
There exists an agent $v$ satisfying $C_0(v).\timer_\KL\ge \tbc/2$ since $C_0\in \KLhalf$.  
From Lemma~\ref{leaderelection:holding:spreads}, the probability that every agent's $\timer_\KL$ becomes more than $0$ in $C_{2m\tau}$ interactions is at least $1-2ne^{-\tau}$.  
From lemma~\ref{leaderelection:holding:count}, the probability that some agent's $\timer_\LF$ becomes less than $\tbc/2$ from its $\timer_\KL$ became more than $0$ until $C_{2m\tau}$ is at least $1-ne^{-\tau}$ since every agent sets $\timer_\LF$ to $\tbc$ during every agent's $\timer_\KL>0$.
Thus, $C_{2m\tau}\in \Fall \cap \Vclean \cap \IDclear \cap \LFqua$ holds.
Let $C_i$ be the configuration where every agent's $\timer_\KL$ becomes $0$ after $C_{2m\tau}$.  
The number of interactions until every agent's $\timer_\KL$ becomes $0$ is less than $2340m\tau\log{n}$ with probability at least $1-e^{-\tau}$ from Lemma~\ref{leaderelection:conv:timeconv} (assigning $\lambda=1$).  
Thus, the probability that $i<2m\tau+2340m\tau\log{n}$ holds is at least $1-e^{-\tau}$ since there is no leader after $C_{2m\tau}$.
Since $v.\timer_\KL>0$ implies $v.\timer_\LF=\tbc$ for any agent $v$, there is some agent $u$ satisfying $u.\timer_\LF=\tbc$ for any $C_{2m\lambda\tau}\ (1\le \lambda \le \lfloor i/(2m\tau)\rfloor)$.
From Lemma~\ref{leaderelection:holding:LFlambda}, the probability that every agent's $\timer_\LF$ does not become $0$ from $C_{2m\tau}$ to $C_i$ is at least $1-1170\log{n}\cdot 3ne^{-\tau}\ge 1-n^{-1}$.  
Note that we assume \assumeN.  
No agent becomes a candidate with that probability.
Thus, $C_{i}\in \KLzero \cap \Fall \cap \Vclean \cap \IDclear$ holds with probability at least $1-3ne^{-\tau}-e^{-\tau}-n^{-1}$ by the union bound.
\end{proof}

\begin{lemma}\label{leaderelection:conv:FalltoSLE}  
Let $C_0\in \KLzero \cap \Fall \cap \Vclean \cap \IDclear$ and $\Xi_\PBC(C_0)=C_0,C_1,\dots$.  
The number of interactions until the configuration reaches $\SLE$ is less than $9360m\tau\log{n}+6m\tau$ with probability  at least $1-4ne^{-\tau}-2e^{-2\tau}-n^{-1}-N^{-1}$.  
\end{lemma}

\begin{proof}  
The number of interactions until an agent's $\timer_\LF$ becomes $0$ is less than $4680m\tau\log{n}$ with probability at least $1-e^{-2\tau}$ from Lemma~\ref{leaderelection:conv:timeconv} (assigning $\lambda=1$).  
Some agents whose $\timer_\LF$ are $0$ become candidates and set their $\timer_\LF$ to $2\tbc$.
When a follower interacts with a candidate, the follower's $\timer_\LF$ is set to $\tbc-1$.
This can be regarded as Larger Time Propagation with candidates to leaders.
From Lemma~\ref{leaderelection:holding:spreads}, the number of interactions until all followers' $\timer_\LF$ become no less than $\tbc/2$ after a candidate appears for the first time is less than $2m\tau$ with probability at least $1-2ne^{-\tau}$.  
Note that we can use Lemma~\ref{leaderelection:holding:spreads} since we can regard candidates as agents whose $\timer_\LF$ are $\tbc$.
Let $p$ be the probability that an agent becomes a new candidate from the $2m\tau$-th interactions after the first candidate appeared until all candidates disappear.  

Next, we analyze the number of interactions from the first candidate appeared until candidates finish generating $\iid$.
A candidate has to interact $\lceil\log{N^2}\rceil$ times to finish generating $\iid$.  
Let $X\sim \text{Bi}(4m\tau,\delta_v/m)$ be a binomial random variables that the number of interactions an agent $v$ interacts during $4m\tau$ interactions.  
From Lemma~\ref{Arisu_upper}, $\Pr(X\le 2\delta_v\tau)\le e^{-\delta_v\tau/2}\le e^{-\tau}$ holds since $\delta_v\ge 2$.  
Since $2\delta_v\tau\ge 4\tau \ge 2\lceil\log{N}\rceil$ holds, it is sufficient to finish generation during $4m\tau$ interactions.  
by the union bound, all agents finish generating $\iid$ during $4m\tau$ interactions with probability at least $1-ne^{-\tau}$.

From the proof of Lemma~\ref{leaderelection:holding:spreads}, the probability that information propagates to all agents during $2m\tau$ interactions is at least $1-ne^{-\tau}$ by the union bound.
Thus, the probability that all agents know the maximum $\iid$ after all candidates finish generating $\iid$ during $2m\tau$ interactions is at least $1-ne^{-\tau}$.  
Thus, the number of interactions until all agents know the maximum $\iid$ after the first candidate appeared is less than $8m\tau$ interactions with  probability at least $1-8ne^{-\tau}-p$.  
The number of times that every candidate decreases $\timer_\LF$ during $8m\tau$ interactions is less than $2\tbc$ with probability at least $1-4ne^{-\tau}$.  
When candidates finish generating $\iid$, they set their $\timer_\LF$ to $2\tbc$.  
Thus, all agents know the maximum $\iid$ before some candidates' $\timer_\LF$ becomes $0$ with that probability.

The number of interactions until every candidate's $\timer_\LF$ becomes $0$ is $4680m\tau\log{n}$ with probability at least $1-e^{-2\tau}$ from Lemma~\ref{leaderelection:conv:timeconv} (assigning $\lambda=2$).  
Let $x$ be a random variable to represent the number of interactions until all candidates disappear after the first candidate appeared.  
From the above, $x$ is less than $6m\tau+4680m\tau\log{n}=2m\tau (3+2340\log{n})$ with probability at least $1-3ne^{-\tau}-e^{-2\tau}-p$.  
From Lemma~\ref{leaderelection:holding:LFlambda}, the probability that all followers' $\timer_\LF$ do not become $0$ during $x$ interactions from the $2m\tau$-th interactions after the first candidate appeared is at least $1-(3+2340\log{n})3ne^{-\tau}\ge 1-1/n$ since \assumeN.  
Thus, $p\le n^{-1}$ holds.

The candidates whose $\timer_\LF$ become $0$ become leaders $\LL_0$.  
We will analyze the probability that there is a unique leader when all candidates disappear.  
From Lemma~\ref{leaderelection:independent}, random numbers that candidates generate are independent and uniform.  
Thus, the probability that multiple candidates generate the same random number is at least $1/2^{\lceil\log{N^2}\rceil}\le 1/N^2$ since random numbers are generated from the range $[2^{\lceil\log{N^2}\rceil},2^{\lceil\log{N^2}\rceil+1})$.  
Therefore, the probability that all $\iid$ that candidates generated are different is at least $1-n/N^2\ge 1-N^{-1}$.

Since candidates always set their $\timer_\mathrm{E}$ to $2\tbc$, a new unique leader's $\timer_\mathrm{E}$ is $2\tbc$ when the leader appears.  
Thus, the configuration belongs to $\Vclean$ and $\Lvzero$.  
Therefore, the number of interactions until the configuration reaches $\Bno \cap \Lone \cap \LFqua \cap  \KLzero \cap  \Lvzero \cap \Vclean$ is less than $4680m\tau\log{n}+6m\tau+4680m\tau\log{n}=9360m\tau\log{n}+6m\tau$ with probability at least $1-e^{-2\tau}-2ne^{-\tau}-ne^{-\tau}-ne^{-\tau}-e^{-2\tau}-n^{-1}-N^{-1}=1-4ne^{-\tau}-2e^{-2\tau}-n^{-1}-N^{-1}$.  
This lemma follows from  $\Bno \cap \Lone \cap \LFqua \cap  \KLzero \cap  \Lvzero \cap \Vclean \subset \SLE$.  
\end{proof}

\begin{lemma}\label{leaderelection:conv:Scol}  
Let $C_0\in \SDcol$ and $\Xi_\PBC(C_0)=C_0,C_1,\dots$.  
The number of interactions required for the configuration to reach $\KLzero \cup \KLhalf$ is less than $2340m\tau\log{n}$ with probability at least $1-e^{-\tau}$.  
\end{lemma}

\begin{proof}  
Firstly, we consider the configuration reaches $\KLzero$.  
If no agent sets $\timer_\KL$ to $\tbc$ before all agents' $\timer_\KL$ become $0$, the number of interactions until all agents' $\timer_\KL$ becomes $0$ is less than $2340m\tau\log{n}$ with probability at least $1-e^{-\tau}$ from Lemma~\ref{leaderelection:conv:timeconv} (assigning $\lambda=1$).  
If some agent sets their $\timer_\KL$ to $\tbc$ during $2340m\tau\log{n}$ interactions, the configuration reaches $\KLhalf$ at the time that some agent sets $\timer_\KL$ to $\tbc$.  
Therefore, the number of interactions until the configuration reaches $\KLzero\cup \KLhalf$ is less than $2340m\tau\log{n}$ steps with probability at least $1-e^{-\tau}$.  
\end{proof}

\begin{lemma}\label{leaderelection:conv:KLzero}
Let $C_0\in \KLzero$ and $\Xi_\PBC(C_0)=C_0,C_1,\dots$.  
The number of interactions until the configuration reaches $\KLzero \cap \Vclean \cup \KLhalf$ is less than $2m\tau+4680m\tau\log{n}$ with  probability at least $1-2ne^{-\tau}-e^{-2\tau}-n^{-1}$.  
\end{lemma}

\begin{proof}  
If $C_0$ belongs to $\Vclean$, $C_0$ already belongs to $\KLzero\cap \Vclean$.  
Thus, we only consider the case $C_0\notin \Vclean$.  
Since there is some agent whose $\timer_\mathrm{V}>0$ and $\timer_\mathrm{E}=2\tbc$ in $C_{0}$, the probability that all agents' $\timer_\mathrm{E}$ become no less than $3\tbc/2$ in $C_{2m\tau}$ is at least $1-2ne^{-\tau}$ from Lemma~\ref{leaderelection:holding:spreads}.  
Let $p$ be the probability that leaders generate a new search virus from $C_{2m\tau}$ before the configuration reaches $\Vclean$.  
Let $x$ be a random variable to represent the number of interactions until all search viruses disappear after $C_{2m\tau}$.
From Lemma~\ref{leaderelection:conv:timeconv} (assigning $\lambda=2$), $x<4680m\tau\log{n}$ with probability at least $1-e^{-2\tau}-p$. 
The probability that all agents' $\timer_\mathrm{E}$ are always no less than $\tbc$ during $C_{2m\tau}$ to $C_{2m\tau+x}$ is at least $1-2340\log{n}\cdot 3ne^{-\tau}\ge 1-1/n$ since $x<4680m\tau\log{n}=2m\tau \cdot 2340\log{n}$, Lemma~\ref{leaderelection:holding:Elambda}, and \assumeN.  
Thus, $p\le n^{-1}$.  
If some agent sets $\timer_\KL$ to $\tbc$ while $2m\tau+x$ interactions, the configuration reaches $\KLhalf$ at that time.  
Therefore, the number of interactions until the configuration reaches $\KLzero \cap \Vclean \cup \KLhalf$ is less than $2m\tau+4680m\tau\log{n}$ steps with probability at least $1-2ne^{-\tau}-e^{-2\tau}-n^{-1}$.  
\end{proof}

\begin{lemma}\label{leaderelection:conv:kill} 
Let $C_0\in \KLzero \cap \Vclean \cap \LFqua \cap  \Ldupl \cap \Bno$ and $\Xi_\PBC(C_0)=C_0,C_1,\dots$.  
The number of interactions until the configuration reaches $\KLhalf$ is less than $14m\tau+10710m\tau\log{n}$ with probability at least $1-23ne^{-\tau}-2e^{-2\tau}-e^{-3\tau/2}-e^{-\tau/2}-e^{-\tau}-4n^{-1}-N^{-2}$.  
\end{lemma}

\begin{sketch}
We will show this lemma by dividing it into two cases: i) $C_0\in \Equa$, and ii) $C_0\notin \Equa$.

Firstly, we analyze the case i) $C_0\in \Equa$.
We show that the configuration reaches $\KLzero\cap \Vclean\cap \LFqua\cap \Ldupl\cap \Bno\cap \Lvzero$ during $3510m\tau\log{n}$ with probability $1-o(1)$.
The number of interactions until some leader's $\timer_\mathrm{E}$ becomes $0$ is less than $4860m\tau\log{n}$ with probability $1-o(1)$.
The number of interactions until some leader finishes generating $\type$ after a leader's $\timer_\mathrm{E}$ becomes $0$ in the first time is less than $6m\tau$ with probability $1-o(1)$.
The number of interactions until the configuration reaches $\KLhalf$ after a leader's $\timer_\mathrm{E}$ becomes $0$ in the first time is less than $8m\tau$ with probability $1-o(1)$.

Secondly, we analyze the case ii) $C_0\notin \Equa$.
Eventually, the configuration reaches $\Equa$ (or reaches $\KLhalf$).
we show that the number of interactions the configuration reaches $\Equa \cap \KLzero \cap \Vclean \cap \LFqua \cap \Ldupl\cap \Bno \cup \KLhalf$ is less than $6m\tau+5850m\tau\log{n}$ with probability $1-o(1)$.
\end{sketch}

\begin{proof}  
Firstly, we consider the case of $C_0\in \Equa$ (\ie, all leaders' $\timer_\mathrm{E}$ is no less than $\tbc/2$).  
Let $p$ be the probability that some leader's $\timer_\mathrm{E}$ becomes $0$ when there exists a leader whose $\timer_\mathrm{E}$ is no less than $\tbc/2$.  
Let $x$ be a random variable to represent the number of interactions until all leaders' $\timer_\mathrm{E}$ become less than $\tbc/2$.
The number of interactions until all leaders' $\timer_\mathrm{E}$ becomes less than $\tbc/2$ is less than $3510m\tau\log{n}$  with probability at least $1-e^{-3\tau/2}-p$ from Lemma~\ref{leaderelection:conv:timeconv} (assigning $\lambda=3/2$).
Note that we can use this lemma since there is no search virus before all leaders' $\timer_{\mathrm{E}}$ become less than $\tbc/2$ with probability $p$.
Thus, $x<3510m\tau\log{n}$ with that probability.
The probability that some leader's $\timer_\mathrm{E}$ becomes $0$ during $C_0$ to $C_x$ is at least $1-3510/2\cdot \log{n}\cdot 3ne^{-\tau}\ge1- 1/n$.  
Thus, $p\le n^{-1}$ since \assumeN.  
Therefore, there is no leader that starts generating a search virus before all agents become $\LL_0$ with high probability.
Since there are leaders before the configuration reaches $\KLhalf$, there exists some agent whose $\timer_\LF=\tbc$.
From Lemma~\ref{leaderelection:holding:LFlambda}, there is no agent whose $\timer_\LF$ becomes $0$ during $C_0$ to $C_x$, and $C_{x}\in \LFqua$ with probability at least $1-3510/2\cdot \log{n}\cdot 3ne^{-\tau}\ge 1/n$ since \assumeN.  
Therefore, $C_x\in \KLzero \cap \Vclean \cap \LFqua \cap \Ldupl \cap \Bno \cap \Lvzero$ with probability at least $1-e^{-3\tau/2}-2n^{-1}$.

Leaders start generating $\type$ when their $\timer_\mathrm{E}$ become $0$, and keep $\timer_\mathrm{E}=2\tbc$ while generating types.
The number of interactions until some leader's $\timer_\mathrm{E}$ becomes $0$ is $4680m\tau\log{n}$ with probability at least $1-e^{-2\tau}$ from Lemma~\ref{leaderelection:conv:timeconv} (assigning $\lambda=2$).  
Let $C'_0$ be the configuration where some leader's $\timer_\mathrm{E}$ becomes $0$ for the first time, and $\Xi'_\PBC(C'_0)=C'_0,C'_1,\dots$.
The probability that all agents' $\timer_\mathrm{E}$ become more than $3\tbc/2$ during $C'_0$ to $C'_{2m\tau}$ is at least $1-2ne^{-\tau}$ from  Lemma~\ref{leaderelection:holding:spreads}.
The probability that all agents' $\timer_\mathrm{E}$ do  not become less than $\tbc$ after their $\timer_\mathrm{E}$ become more than $3\tbc/2$ is at least $1-ne^{-\tau}$ from Lemma~\ref{leaderelection:holding:count}.
Thus, after $C'_{2m\tau}$, all leaders do not start generating $\type$ while their $\timer_\mathrm{E}$ are not $0$.

We analyze the number of interactions until a leader finishes generating $\type$.
We only consider the case there is no leader $\LL_1$.
Note that from the above, all $\LL_1$ became $\LL_0$ with high probability.
Every leader has to interact $\lceil\log{N}\rceil$ times to finish generating $\type$.  
Let $X\sim \text{Bi}(4m\tau,\delta_v/m)$ be a binomial random variable that represents the number of interactions that a leader $v$ interacts during $4m\tau$ interactions.  
From Lemma~\ref{Arisu_upper}, $\Pr(X\le 2\delta_v\tau)\le e^{-4\delta_v\tau/8}\le e^{-\tau}$.  
Thus, a leader $v$ interacts no less than $2\delta_v\tau>\lceil\log{N}\rceil$ times with probability at least $1-e^{-\tau}$.  
Therefore, the probability that some leader finishes generating $\type$ during $C'_{0}$ to $C'_{6m\tau}$ is at least $1-e^{-\tau}$. 
When a leader finishes generating $\type$, the leader sets their $\timer_\mathrm{V}$ to $2\tbc$.
Let $C'_i$ be the configuration that a leader finishes generating $\type$ for the first time.
The generated search virus spreads to all agent during $2m\tau$ interactions with probability at least $1-2ne^{-\tau}$ from Lemma~\ref{leaderelection:holding:spreads}.
Specifically, since there is some leader $v$ satisfying $C'_i(v).\timer_\mathrm{V}=2\tbc$, every agent's $\timer_\mathrm{V}$ becomes more than $3\tbc/2$ during $C'_i$ to $C'_{i+2m\tau}$ with that probability.
When the search virus reaches other leader, we consider three cases:
i) the leader is $\LL_0$, ii) the leader is $\LL_1$ and generating $\type$, iii) the leader is $\LL_1$ and finished generating $\type$.
In the cases i) and ii), the configuration becomes $\KLhalf$.
In the case iii), the configuration becomes $\KLhalf$ if the leader has alive search virus and their types are different.
The probability that the other leader's search virus is alive when the first generated search virus reaches the leader is at least $1-3ne^{-\tau}$ from Lemma~\ref{leaderelection:holding:count} and union bound since this event happens from $C'_i$ to $C'_{i+2m\tau}$ with probability at least $1-2ne^{-\tau}$.
We analyze the probability that some $\type$s are different.
Leaders generate $\type$ uniformly and independently from the range of $[2^{\lceil\log{N}\rceil},2^{\lceil\log{N}\rceil+1})$ from Lemma~\ref{leaderelection:independent}.  
Thus, the probability that two $\type$s are the same is no more than $(1/2^{\lceil\log{N}\rceil})^2\le N^{-2}$.  
Therefore, the probability that at least two search viruses are different is at least $1-N^{-2}$.  
Therefore, in the case iii), the configuration becomes $\KLhalf$ with probability at least $1-3ne^{-\tau}-N^{-2}$.
The probability that all agents' $\timer_\LF$ do not become $0$ during $C'_0$ to $C'_{8m\tau}$ is at least $1-12ne^{-\tau}$ from Lemma~\ref{leaderelection:holding:LFlambda}.

From the above, if $C_0\in \Equa$, the number of interactions until the configuration reaches $\KLhalf$ is less than $8m\tau+4680m\tau\log{n}$ with probability at least  $1-18ne^{-\tau}-e^{-2\tau/}-e^{-3\tau/2}-e^{-\tau}-2n^{-1}-N^{-2}$ by the union bound.

Next, we consider $C_0\notin \Equa$.  
The number of interactions until some leaders' $\timer_\mathrm{E}$ becomes $0$ is less than $1170m\tau\log{n}$ with probability at least $1-e^{-\tau/2}$.  
When some leader's $\timer_\mathrm{E}$ becomes $0$, the leader starts generating $\type$.  
Let $C''_0$ be the configuration where some leader's $\timer_\mathrm{E}$ becomes $0$ for the first time, and $\Xi''_\PBC(C''_0)=C''_0,C''_1,\dots$.  
After $C''_{2m\tau}$, all agents' $\timer_\mathrm{E}$ is always more than $3\tbc/2$ with probability at least $1-3ne^{-\tau}$ from lemma~\ref{leaderelection:holding:spreads} since there is an agent satisfying $\timer_\mathrm{E}=2\tbc$ in $C''_0$.
Thus, $\LL_1$ does not become $\LL_0$ after $C''_{2m\tau}$ until all search viruses disappear with that probability.  
When all search viruses disappear, the configuration belongs to $\KLzero \cap \Vclean \cap \LFqua \cap \Equa \cap \Ldupl \cap \Bno$ with high probability.
The probability that all agents’ $\timer_\LF$ do not become 0 until all search virus disappear after a leader became $\LL_1$ in the first time is at least $1-3(3+2340\log{n})ne^{-\tau} \ge 1-n^{-1}$.
Also, the probability that all leaders’ $\timer_\mathrm{E}$ do not become 0 until all search virus disappear after a leader became $\LL_1$ in the first time is at least $1-3(3+2340\log{n})ne^{-\tau} \ge 1-n^{-1}$.
Therefore, the number of interactions until the configuration reaches $\KLzero \cap \Vclean \cap \LFqua \cap \Equa \cap \Ldupl \cap \Bno$ from $C_0$
is less than $6m\tau+5850m\tau\log{n}$ with probability at least $1-5ne^{-\tau}-e^{-2\tau}-e^{-\tau/2}-2n^{-1}$.

From the above, the number of interactions that the configuration reaches $\KLhalf$ is $14m\tau+10710m\tau\log{n}$ with probability at least $1-23ne^{-\tau}-2e^{-2\tau}-e^{-3\tau/2}-e^{-\tau/2}-e^{-\tau}-4n^{-1}-N^{-2}$.
\end{proof}

\begin{figure}[!htb]
    \centering
    \includegraphics[scale=0.4]{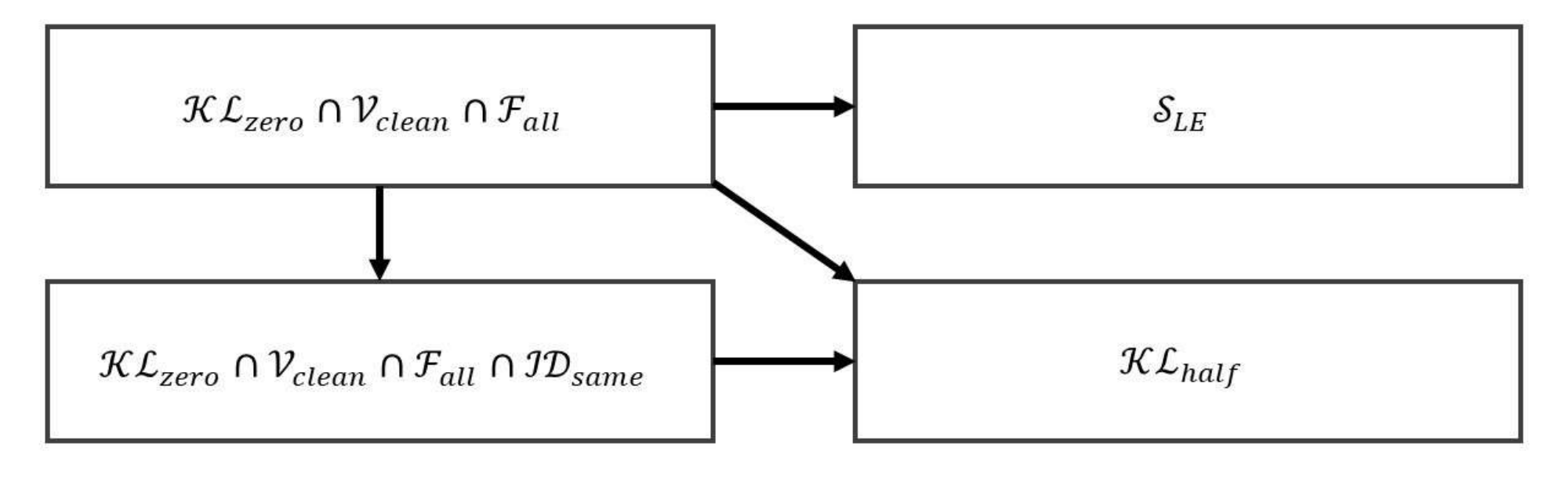}
    \caption{Structure of the proof if $C_0\in \KLzero\cap \Vclean\cap \Fall$ in Lemma~\ref{leaderelection:conv:KLzeroVclean}.}
    \label{figure:allF}
\end{figure}

\begin{figure}[!htb]
    \centering
    \includegraphics[scale=0.4]{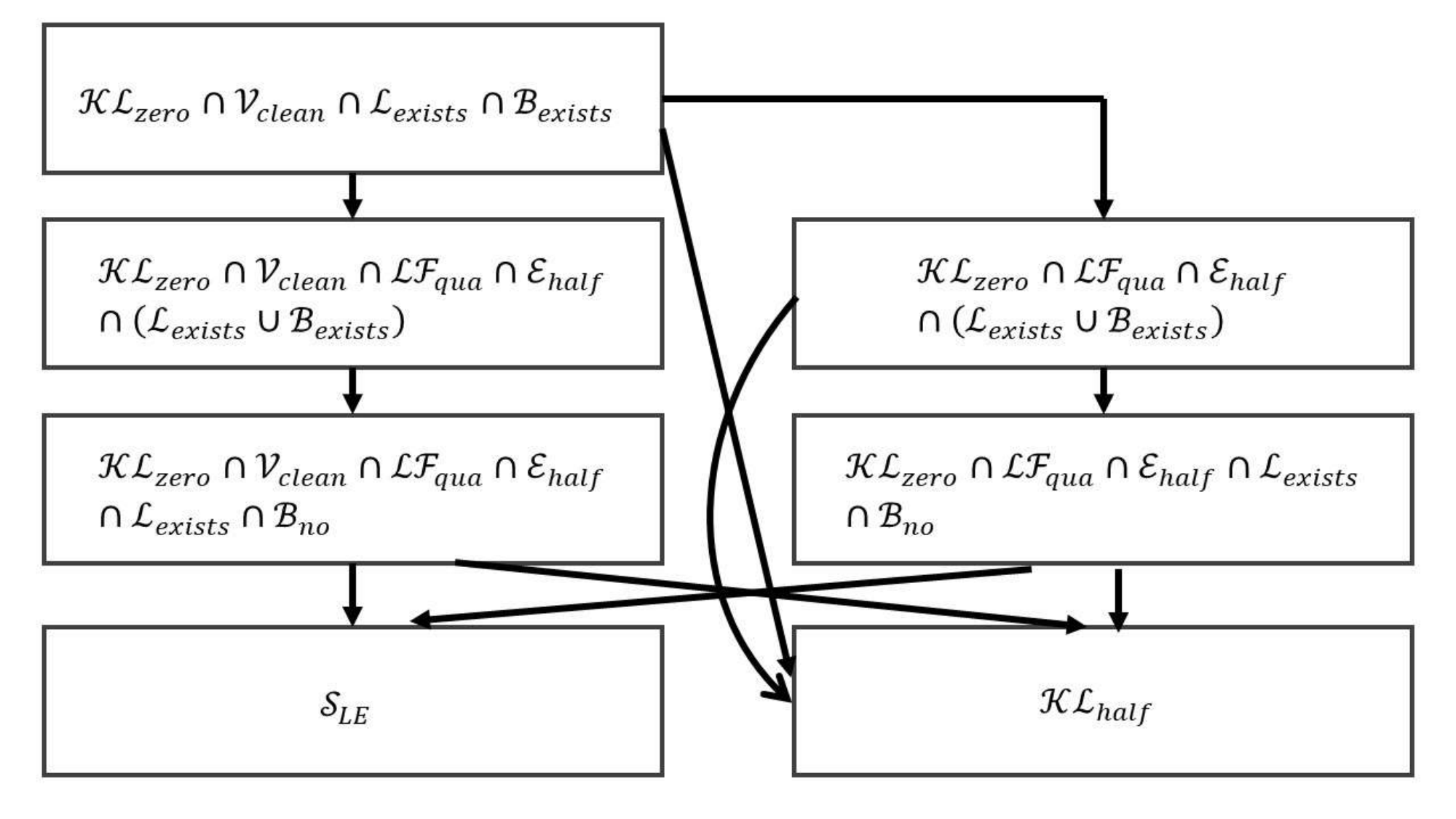}
    \caption{Structure of the proof if $C_0\in \KLzero\cap \Vclean\cap \Lexists \cap \Bexists$ in Lemma~\ref{leaderelection:conv:KLzeroVclean}.}
    \label{figure:BandL}
\end{figure}

\begin{figure}[!htb]
    \centering
    \includegraphics[scale=0.4]{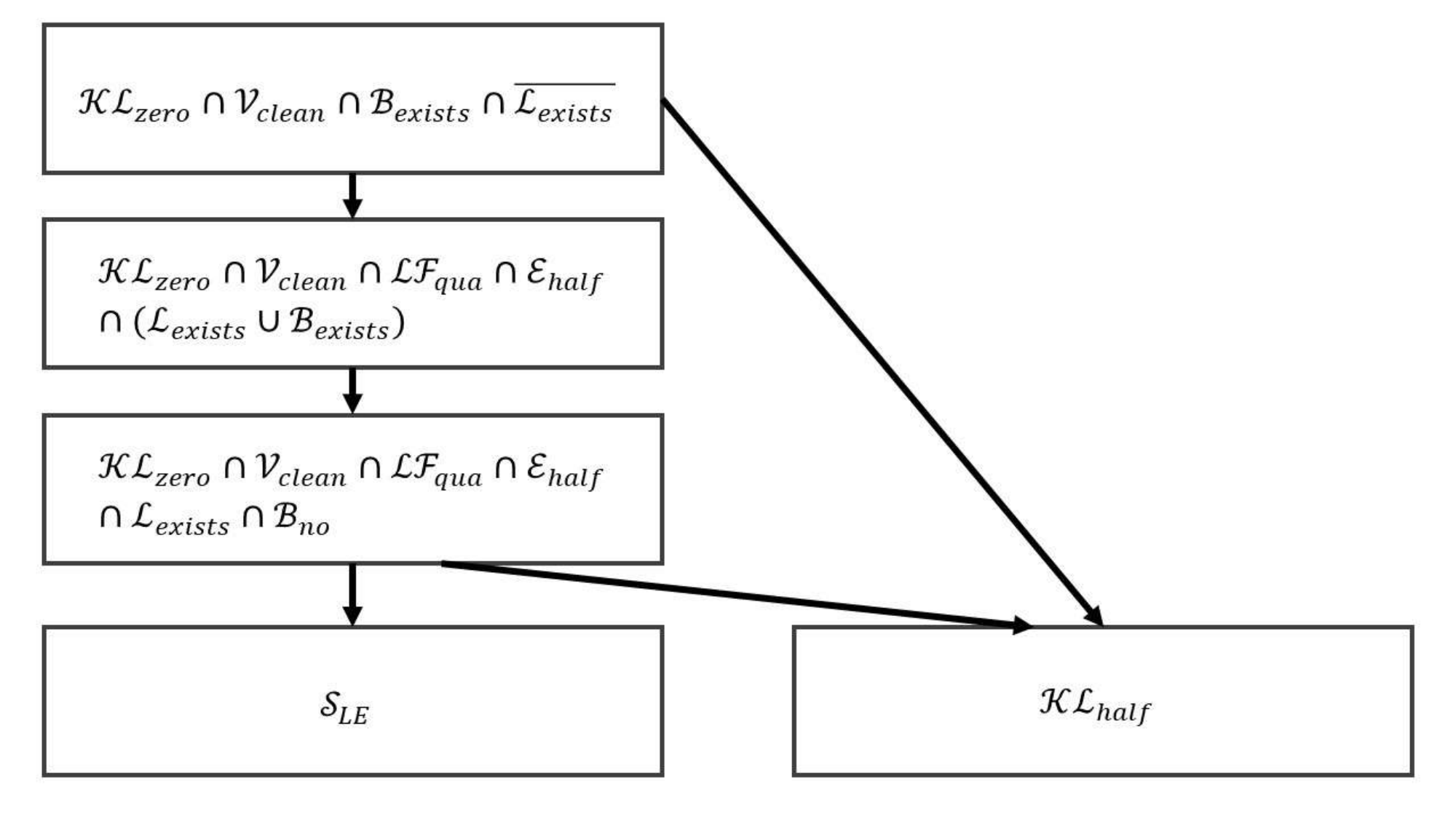}
    \caption{Structure of the proof if $C_0\in \KLzero\cap \Vclean\cap \Bexists \cap \overline{\Lexists}$ in Lemma~\ref{leaderelection:conv:KLzeroVclean}.}
    \label{figure:onlyB}
\end{figure}

\begin{figure}[!htb]
    \centering
    \includegraphics[scale=0.4]{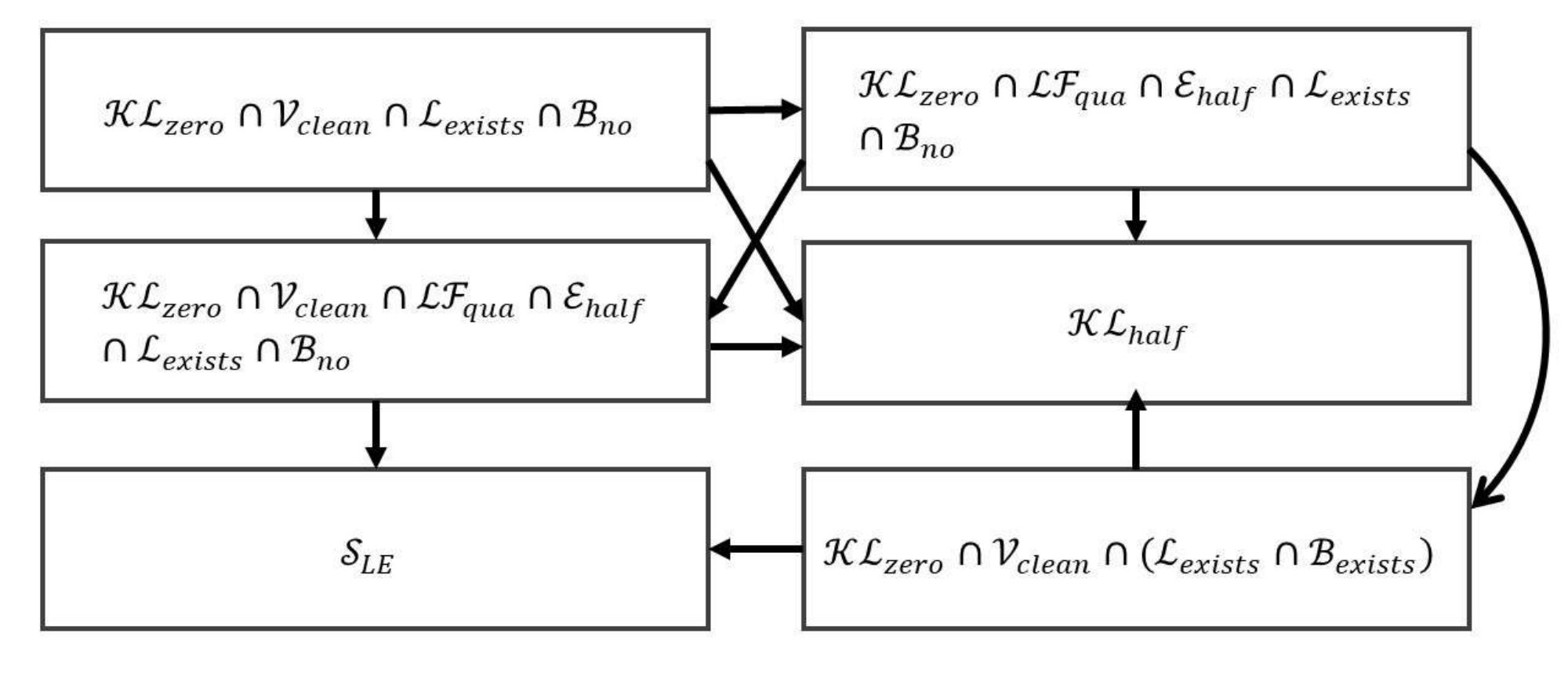}
    \caption{Structure of the proof if $C_0\in \KLzero\cap \Vclean\cap \Lexists \cap \Bno$ in Lemma~\ref{leaderelection:conv:KLzeroVclean}.}
    \label{figure:onlyL}
\end{figure}

\begin{lemma}\label{leaderelection:conv:KLzeroVclean}  
Let $C_0\in \KLzero \cap \Vclean$, and $\Xi_\PBC(C_0)=C_0,C_1,\dots$.  
The number of interactions until the configuration reaches $\SLE \cup \KLhalf$ is less than $4m\tau+18720m\tau\log{n}$ with probability  $1-o(1)$.   
\end{lemma}

\begin{sketch}
We will show this lemma by dividing it into four cases:  
i) the case where there is no leader and no candidate in $C_0$, ii) the case where there are only candidates in $C_0$, iii) the case where there are only leaders in $C_0$, and iv) the case where there are leaders and candidates in $C_0$.

Firstly, we show the case i) $C_0\in \KLzero \cap \Vclean \cap \Fall$ as shown in Figure~\ref{figure:allF}.
Specifically, we show that the number of interactions until the configuration reaches $\KLhalf \cup \SLE$ is less than $8m\tau+9360m\tau\log{n}$ with probability $1-o(1)$.

Secondly, we show the case ii) $C_0\in \KLzero \cap \Vclean \cap \Lexists \cap \Bexists$ as shown in Figure~\ref{figure:BandL}.
Specifically, we show that the number of interactions until the configuration reaches $\KLhalf \cup \SLE$ is less than $2m\tau+18720m\tau\log{n}$ with probability $1-o(1)$.

Thirdly, we show the case iii) $C_0\in \KLzero \cap \Vclean \cap \Bexists \cap \overline{\Lexists}$ as shown in Figure~\ref{figure:onlyB}.
Specifically, we show that the number of interactions until the configuration reaches $\KLhalf \cup \SLE$ is less than $20m\tau+14040m\tau\log{n}$ with probability $1-o(1)$.

Finally, we show the case iv) $C_0\in \KLzero \cap \Vclean \cap \Lexists \cap \Bno$ as shown in Figure~\ref{figure:onlyL}.
Specifically, we show that the number of interactions until the configuration reaches $\KLhalf \cup \SLE$ is less than $4m\tau+18720m\tau\log{n}$ with probability $1-o(1)$.
\end{sketch}

\begin{proof}  
We show this lemma by dividing it into four cases:  
the case where there is no leader and no candidate in $C_0$, the case where there are only candidates in $C_0$, the case where there are only leaders in $C_0$, and the case where there are leaders and candidates in $C_0$.

Firstly, we analyze when there is no leader and no candidate in $C_0$.  
Note that $C_0\in \KLzero \cap \Vclean \cap \Fall$.  
There are some followers whose $\timer_\LF$ becomes $0$ after $C_0$.  
At those times, if the follower's $\iid$ is not $1$, the phase moves to Global Reset, and the configuration belongs to $\KLhalf$.  
The number of interactions until this event  happens is less than $4680m\tau\log{n}$ with probability $1-o(1)$ from Lemma~\ref{leaderelection:conv:timeconv} (assigning $\lambda=2$).  
We consider the case where Global Reset did not happen.  
The number of interactions until for the $\iid$ generation finishes after the first follower becomes a candidate is less than $6m\tau$ interactions with probability $1-o(1)$ from the proof of Lemma~\ref{leaderelection:conv:FalltoSLE}.  
If the maximum $\iid$ of candidates is no less than the maximum $\iid$ of all agents, the configuration belongs to $\SLE$ while $8m\tau$ interactions after the first follower becomes a candidate with probability $1-o(1)$ from the proof of Lemma~\ref{leaderelection:conv:FalltoSLE}.  
Otherwise, all agents become followers; thus, the configuration belongs to $\KLzero \cap \Vclean \cap \Fall \cap \IDsame \cap \overline{\IDclear}$ while $8m\tau$ interactions after the first follower becomes a candidate with probability $1-o(1)$ from the proof of Lemma~\ref{leaderelection:conv:FalltoSLE}.  
The number of interactions until the configuration reaches $\KLhalf$ after the configuration belonged to $\KLzero \cap \Vclean \cap \Fall \cap \IDsame \cap \overline{\IDclear}$ is less than $4680m\tau\log{n}$ with probability $1-o(1)$ from Lemma~\ref{leaderelection:conv:timeconv} (assigning $\lambda=2$), since all agents' $\iid \ne 1$ holds.  
Therefore, if $C_0\in \Fall$, the number of interactions until the configuration reaches $\SLE \cup \KLhalf$ is less than $8m\tau+9360m\tau\log{n}$ with probability $1-o(1)$.

Secondly, we analyze when there are leaders and candidates.  
Note that $C_0\in \KLzero \cap \Vclean \cap \Lexists \cap \Bexists$.  
We consider after $2m\tau$ interactions from $C_0$, divided into two cases:  
i) the case where the search virus is not created, and ii) the case where the search virus is created.
In both cases, when kill virus is created, the configuration reaches $\KLhalf$.
i) If the search virus is not created, $C_{2m\tau}\in \KLzero \cap \Vclean \cap(\Lexists \cup \Bexists)\cap \LFqua \cap \Ehalf$ holds with probability $1-o(1)$ from Lemma~\ref{leaderelection:holding:spreads}.  
Eventually, candidates become leaders or disappear; thus, the configuration will become $\KLzero \cap \Vclean \cap \Lexists \cap \LFqua \cap \Ehalf \cap \Bno$.  
The number of interactions until the configuration reaches that is less than $4680m\tau\log{n}$ with probability $1-o(1)$ from Lemma~\ref{leaderelection:conv:timeconv} (assigning $\lambda=2$).
If there are multiple leaders, from Lemma~\ref{leaderelection:conv:kill}, the phase moves to Global Reset during $18m\tau+9360m\tau\log{n}$ with probability $1-o(1)$.  
If there is a leader and no candidate, the configuration belongs to $\SLE$.
ii) We consider if the search virus was created.  
If the phase did not move to Global Reset during $2m\tau$ interactions, $C_{2m\tau} \in \KLzero \cap (\Lexists \cup \Bexists) \cap \LFqua \cap \Ehalf$ holds with probability $1-o(1)$ from Lemma~\ref{leaderelection:holding:spreads}.  
The number of interactions until all candidates disappear is less than $4680m\tau\log{n}$ with probability $1-o(1)$ from Lemma~\ref{leaderelection:conv:timeconv} (assigning $\lambda=2$).  
After all candidates disappear, if there is no search virus, the configuration reaches $\KLhalf$ or $\SLE$ during $4680m\tau\log{n}$ interactions, with probability $1-o(1)$ from the above.
Otherwise, after all search viruses disappear, the configuration moves to Global Reset if there are multiple leaders.
Note that the configuration belongs to $\KLzero\cap \Vclean\cap \Ehalf \cap \Bno\cap \Ldupl \cap \LFqua$.
Thus, the number of interactions until this event happens is $9360m\tau\log{n}$ with probability $1-o(1)$ from the above, Lemma~\ref{leaderelection:holding:LFlambda}, and Lemma~\ref{leaderelection:holding:Elambda}.
If there is a leader, the configuration belongs to $\SLE$.
Therefore, if $C_0\in \Lexists \cap \Bexists$, the number of interactions until the configuration reaches $\KLhalf\cup \SLE$ is less than $2m\tau+18720m\tau\log{n}$ with probability $1-o(1)$.  
Note that $\timer_\LF$ did not become $0$, with probability $1-o(1)$ from Lemma~\ref{leaderelection:holding:LFlambda}.

Thirdly, we analyze when there are only candidates. 
Note that $C_0\in \KLzero \cap \Vclean \cap \Bexists\cap \overline{\Lexists}$.  
$C_{2m\tau}\in \KLzero \cap \Vclean \cap (\Bexists \cup \Lexists) \cap \LFqua \cap \Ehalf \cup \KLhalf$ holds with probability $1-o(1)$ from Lemma~\ref{leaderelection:holding:spreads}, since some followers' $\timer_\LF$ became $0$.  
Note that the phase moves to Global Reset when followers whose $\iid\ne 1$ become candidates.  
If $C_{2m\tau}\in \KLzero \cap \Vclean \cap (\Bexists \cup \Lexists) \cap \LFqua \cap \Ehalf$ holds, the number of interactions until all candidates disappears after $C_{2m\tau}$ is less than $4680m\tau\log{n}$ with probability $1-o(1)$ from Lemma~\ref{leaderelection:conv:timeconv} (assigning $\lambda=2$).  
Thus, the configuration belongs to $\KLzero \cap \Vclean \cap \Lexists \cap \LFqua \cap \Ehalf\cap \Bno$.  
If there exists a unique leader, the configuration belongs to $\SLE$; otherwise, after $18m\tau+9360m\tau\log{n}$ interactions, the configuration will belong to $\KLhalf$ with probability $1-o(1)$ from Lemma~\ref{leaderelection:conv:kill}.  
Therefore, if $C_0\in \KLzero \cap \Vclean \cap \Bexists\cap \overline{\Lexists}$, the number of interactions until the configuration reaches to $\KLhalf\cup \SLE$ is less than $20m\tau+14040m\tau\log{n}$ with probability $1-o(1)$.  
Note that $\timer_\LF$ did not become $0$, with probability $1-o(1)$ from~\ref{leaderelection:holding:LFlambda}.

Finally, we analyze when there are only leaders.  
Note that $C_0\in \KLzero \cap \Vclean \cap \Lexists \cap \Bno$.  
After $2m\tau$ interactions, i) if there are candidates and no search virus, $C_{2m\tau}$ belongs to $\KLzero \cap \Vclean \cap(\Lexists \cap \Bexists)\cap \LFqua \cap \Ehalf$ with probability $1-o(1)$ from Lemma~\ref{leaderelection:holding:spreads}.  
From the above, this configuration reaches  $\KLhalf$ during $18m\tau+14040m\tau\log{n}$ interactions with probability $1-o(1)$.  
ii) If there are candidates and search viruses, $C_{2m\tau}$ belongs to $\KLzero\cap \Vclean \cap(\Lexists \cap \Bexists)\cap \LFqua \cap \Ehalf$, and the configuration reaches to $\KLhalf$ during $2m\tau+18720m\tau\log{n}$ interactions with probability $1-o(1)$ from the above.  
iii) We consider if there are no candidates and no search viruses, $C_{2m\tau}$ belongs to $\KLzero \cap \Lexists \cap \LFqua \cap \Ehalf \cap \Bno$.  
If the number of leaders is $1$, the configuration belongs to $\SLE$; otherwise, the configuration will belong to $\KLhalf$ during $18m\tau+9360m\tau\log{n}$ steps with probability $1-o(1)$ from Lemma~\ref{leaderelection:conv:kill}.  
iv) We consider if there are no candidates but there are some search viruses; $C_{2m\tau}$ belongs to $\KLzero \cap \Lexists \cap \LFqua \cap \Ehalf \cap \Bno$.  
The number of interactions until all search viruses  disappear is less than $4860m\tau\log{n}$ with probability $1-o(1)$ from Lemma~\ref{leaderelection:conv:timeconv} (assigning $\lambda=2$).  
After that, the configuration belongs to $\KLzero \cap \Lexists \cap \LFqua \cap \Ehalf \cap \Bno$.  
From the above, the configuration will reach $\SLE \cup \KLhalf$ during $20m\tau+9360m\tau\log{n}$ interactions with probability $1-o(1)$.  
Therefore, if $C_0\in \KLzero \cap \Vclean \cap \Lexists \cap \Bno$, the number of interactions until the configuration reaches $\KLhalf\cup \SLE$ is less than $4m\tau+18720m\tau\log{n}$ with probability $1-o(1)$.  
Note that $\timer_\LF$ did not become $0$, with probability $1-o(1)$ from Lemma~\ref{leaderelection:holding:LFlambda}.

Overall, this lemma follows.  
\end{proof}

\LEconvconv*

\begin{proof}  
From Lemma~\ref{leaderelection:conv:Scol}, Lemma~\ref{leaderelection:conv:KLzero}, Lemma~\ref{leaderelection:conv:KLzeroVclean}, Lemma~\ref{leaderelection:conv:KLhalf}, and Lemma~\ref{leaderelection:conv:FalltoSLE},  
$\Pr(\exists i\in O(m\tau\log{n}):C_i\in \SLE)=1-o(1)=\Omega(1)$ holds.  
Thus, this lemma follows from Lemma~\ref{leaderelection:conv:katei}.  
\end{proof}

\ndteiri*

\begin{proof}  
If we use $\PLRU$ as a self-stabilizing two-hop coloring protocol, from Lemma~\ref{twohoprandom:conv}, $\max_{C\in \Call(\PBC)}\ECT_\PBC(C,\Scol)=O(mn)$ holds with high probability.  
From Lemma~\ref{leaderelection:conv:conv}, $\max_{C\in \Scol}\ECT_\PBC(C,\SLE)=O(m\tau\log{n})$ holds with high probability, and from Lemma~\ref{leaderelection:holding:hold}, $\min_{C\in \SLE}\EHT_\PBC(C,LE)=\Omega(\tau e^\tau)$ holds when $\tau \ge \max(2d, 2^{-1}\lceil\log{N}\rceil,15+3\log{n})$.  
Therefore, this theorem holds.  
\end{proof}

When we assign $\tau=2N+15$, the convergence time becomes $O(mN\log{n})$, and the holding time becomes $\Omega(Ne^{2N})$.

\dteiri*

\begin{proof}  
If we use $\PDLRU$ as a self-stabilizing two-hop coloring protocol, from Lemma~\ref{twohop:allconv}, $\max_{C\in \Call(\PBC)}\ECT_\PBC(C,\Scol)=O(m(n+\Delta\log{N}))$ holds with high probability.  
From Lemma~\ref{leaderelection:conv:conv}, $\max_{C\in \Scol}\ECT_\PBC(C,\SLE)=O(m\tau\log{n})$ holds with high probability, and from Lemma~\ref{leaderelection:holding:hold}, $\min_{C\in \SLE}\EHT_\PBC(C,LE)=\Omega(\tau e^\tau)$ holds when $\tau \ge \max(2d, 2^{-1}\lceil\log{N}\rceil,15+3\log{n})$.  
Therefore, this theorem holds.  
\end{proof}

When we assign $\tau=2N+15$, convergence time of randomized and deterministic protocols become $O(mN\log{n})$ and $O(mN\log{N})$ respectively, and holding time of both protocols becomes $\Omega(Ne^{2N})$.

\section{Probabilistic Tools}
In this paper, we use the following lemmas on inequalities for the proofs.

\begin{lemma}[Eq. 4.2 in \cite{Arisu}]\label{Arisu_lower}
For any binomial random variable $X$ and $0< \kappa \le 1$:
\[\Pr(X\ge (1+\kappa)E[X])\le e^{-\kappa^2E[X]/3}\]
\end{lemma}

\begin{lemma}[Eq. 4.5 in \cite{Arisu}]\label{Arisu_upper}
For any binomial random variable $X$ and $0< \kappa \le 1$:
\[\Pr(X\le (1-\kappa)E[X])\le e^{-\kappa^2E[X]/2}\]
\end{lemma}

\begin{lemma}[Theorem 2.1 in Janson~\cite{JANSON20181}]\label{Janson_lower}
For $k\in \mathbb{N}$ and $i=1,\dots,k$,
let $X_i=\text{Geom}(p_i)$ be an independent geometric random variable and $X=\sum_{i=1}^{k}X_i$ where $0<p_i\le 1$.
Define $p=\min_{i\in [1,k]}p_i$, and $\lambda \ge 1$:
\[\Pr(X\ge \lambda E[X])\le e^{-pE[X](\lambda-1-\log_e{\lambda})}\]
\end{lemma}

\end{document}